\PassOptionsToPackage{dvipsnames}{xcolor}
\documentclass[onefignum,onetabnum]{siamart171218}
\usepackage[top=1.4in,bottom=1.275in,left=1.4in,right=1.4in]{geometry}

\usepackage[T1]{fontenc}
\usepackage[english]{babel} \usepackage{csquotes} \usepackage{microtype} \usepackage{hyphenat} \usepackage{siunitx}
\usepackage{amsfonts}
\usepackage{amssymb}
\usepackage{amscd}
\usepackage{mathtools}
\usepackage[mathscr]{eucal}
\usepackage{calligra}
\DeclareMathAlphabet{\mathcalligra}{T1}{calligra}{m}{n}
\DeclareFontShape{T1}{calligra}{m}{n}{<->s*[1.1]callig15}{}
\usepackage{bm}
\usepackage{autonum}\usepackage[nice]{nicefrac}

\usepackage{comment}
\usepackage{graphicx}
\usepackage[colorinlistoftodos,prependcaption,textsize=small]{todonotes}

\usepackage{cleveref}
\usepackage{booktabs}
\usepackage{multirow}
\usepackage{environ}
\usepackage{tikz}
\usetikzlibrary{external}
\usepackage{pgfplots}
\usepackage{pgfplotstable}
\usepackage[colorinlistoftodos,prependcaption,textsize=small]{todonotes}
\usepackage{csvsimple}
\usepackage{algorithmic}

\crefname{equation}{}{}

\pgfplotsset{compat=1.14}
\usetikzlibrary{quotes,arrows.meta,patterns,decorations.markings,calc,plotmarks,pgfplots.groupplots}
\ifpdf
\DeclareGraphicsExtensions{.eps,.pdf,.png,.jpg}
\else
\DeclareGraphicsExtensions{.eps}
\fi

\makeatletter
\newsavebox{\measure@tikzpicture}
\NewEnviron{scaletikzpicturetowidth}[1]{	\def\tikz@width{#1}	\begin{lrbox}{\measure@tikzpicture}		\BODY
	\end{lrbox}	\pgfmathparse{#1/\wd\measure@tikzpicture}		\BODY
}
\makeatother

\newcommand{\revised}[1]{#1}

 \renewcommand{\v}[1]{\ensuremath{\mathbf{#1}}} \newcommand{\gv}[1]{\ensuremath{\mbox{\boldmath$ #1 $}}}

 \newcommand{\scpd}[3][]{\left< #2, #3 \right>_{#1}}     \newcommand{\grad}[1]{\gv{\nabla} #1}  \renewcommand{\div}[1]{\gv{\nabla} \cdot #1}  \newcommand{\curl}[1]{\gv{\nabla} \times #1} \let\baraccent=\= \renewcommand{\=}[1]{\stackrel{#1}{=}}

\theoremstyle{definition}

\theoremstyle{remark}

\newcommand{\integral}[4]{\int\limits_{#1}^{#2} #3 \,\mathrm{d}#4}
\newcommand{\integrald}[4]{\int_{#1}^{#2} #3 \,\mathrm{d}#4}
\newcommand{\reals}{\mathbb{R}}
\newcommand\restr[2]{{		\left.\kern-\nulldelimiterspace 		#1 		\vphantom{\big|} 		\right|_{#2} }}
\newcommand{\identity}{I}
\newcommand{\trace}[1]{\mathrm{tr}(#1)} \newcommand{\stress}{\bm{\sigma}}
\newcommand{\strain}{\bm{\varepsilon}}
\newcommand{\normal}{\mathbf{n}}
\newcommand{\ltwonorm}[1]{\|{#1}\|_2}
\newcommand{\eucnorm}[1]{\|{#1}\|_2}

\newcommand{\ttme}{\texttt{me}}
\newcommand{\ttmnw}{\texttt{mnw}}
\newcommand{\ttmn}{\texttt{mn}}
\newcommand{\ttts}{\texttt{ts}}
\newcommand{\tttse}{\texttt{tse}}
\newcommand{\tttw}{\texttt{tw}}
\newcommand{\tttc}{\texttt{tc}}
\newcommand{\ttbc}{\texttt{bc}}
\newcommand{\ttbe}{\texttt{be}}
\newcommand{\ttbnw}{\texttt{bnw}}
\newcommand{\ttbn}{\texttt{bn}}
\newcommand{\ttms}{\texttt{ms}}
\newcommand{\ttmse}{\texttt{mse}}
\newcommand{\ttmw}{\texttt{mw}}

\newcommand{\tet}{T}

\newcommand{\microtet}{t}

\newcommand{\localstiff}{a}

\newcommand{\stencil}{S}
\newcommand{\refstencil}{\hat{S}}
\newcommand{\corrterm}{R}

\newcommand{\bfe}{\bm{e}}
\newcommand{\bfx}{\bm{x}}
\newcommand{\bfb}{\bm{b}}

\newcommand{\mesh}{\mathcal{T}}

\newcommand{\coeff}{k}

\tikzset{->-/.style={decoration={
			markings,
			mark=at position #1 with {\arrow{>}}},postaction={decorate}}}
		
\theoremstyle{remark}
\newtheorem{remark}[theorem]{Remark}

\newcommand{\sff}{\mathsf{f}}

\newcommand{\sfk}{\mathsf{k}}

\newcommand{\sfp}{\mathsf{p}}

\newcommand{\sfu}{\mathsf{u}}

\newcommand{\sfx}{\mathsf{x}}
\newcommand{\sfy}{\mathsf{y}}

\newcommand{\sfA}{\mathsf{A}}
\newcommand{\sfB}{\mathsf{B}}
\newcommand{\sfC}{\mathsf{C}}

\newcommand{\sfzero}{\mathsf{0}}

\definecolor{color1}{rgb}{0, 0.4470, 0.7410}
\definecolor{color2}{rgb}{0.8500, 0.3250, 0.0980}
\definecolor{color3}{rgb}{0.9290, 0.6940, 0.1250}
\definecolor{color4}{rgb}{0.4940, 0.1840, 0.5560}
\definecolor{color5}{rgb}{0.4660, 0.6740, 0.1880}
\definecolor{color6}{rgb}{0.3010, 0.7450, 0.9330}
\definecolor{color7}{rgb}{0.6350, 0.0780, 0.1840}

\DeclareMathOperator*{\argmin}{arg\,min}

\newcommand{\TheTitle}{Stencil scaling for vector-valued PDEs on hybrid grids with applications to generalized Newtonian fluids}
\newcommand{\TheAuthors}{Daniel~Drzisga, Ulrich~R\"ude, and Barbara~Wohlmuth}

\title{{\TheTitle}\thanks{Submitted to the editors \today.
			}}

\author{
Daniel~Drzisga\thanks{Institute for Numerical Mathematics (M2), Technische
	Universit{\"a}t M{\"u}nchen}
\and Ulrich~R\"ude\thanks{Dept.~of Computer Science 10,
	Friedrich-Alexander-Universit{\"a}t Erlangen-N{\"u}rnberg}
\and Barbara~Wohlmuth\footnotemark[2]}

\ifpdf
\hypersetup{
  pdftitle={\TheTitle},
  pdfauthor={\TheAuthors}
}
\fi

\begin{document}

\maketitle

\begin{abstract}
Matrix-free finite element implementations for large applications provide an attractive alternative to standard sparse matrix data formats due to the significantly reduced memory consumption.
Here, we show that they are also competitive with respect to the run time \revised{in the low order case} if combined with suitable stencil scaling techniques.
We focus on variable coefficient vector-valued partial differential equations as they arise in many physical applications.
The presented method is based on scaling constant reference stencils \revised{originating from a linear finite element discretization} instead of evaluating the bilinear forms on-the-fly.
\revised{This method assumes the usage of hierarchical hybrid grids, and it may be applied to vector-valued second-order elliptic partial differential equations directly or as a part of more complicated problems.}
We provide theoretical and experimental performance estimates showing the advantages of this new approach compared to the traditional on-the-fly integration and stored matrix approaches.
In our numerical experiments, we consider two specific mathematical models.
Namely, linear elastostatics and incompressible Stokes flow.
The final example considers a non-linear shear-thinning generalized Newtonian fluid.
For this type of non-linearity, we present an efficient approach to compute a regularized strain rate which is then used to define the node-wise viscosity.
\revised{Depending on the compute architecture, we could observe maximum speedups of $64$\% and $122$\% compared to the on-the-fly integration.
The largest considered example involved solving a Stokes problem with $12\,288$ compute cores on the state of the art supercomputer SuperMUC-NG.}

\end{abstract}



\section{Introduction}
\label{sec:introduction}
In this article, we study the efficiency of large scale low-order finite element computations
and we examine which accuracy can be obtained at what cost.
High performance computing is expensive, not only in terms of
investments in supercomputer systems, but also in terms of operational cost. 
In particular,
energy consumption is becoming a critical factor; see, e.g.,~the emerging rankings like the GREEN500 list\footnote{\url{https://www.top500.org/green500/lists/2019/11/}}.
Therefore, it is crucial to rethink long-established computing practices 
and to study, quantify, and improve the efficiency of current numerical algorithms.

We primarily strive to reduce the absolute compute times.
This is of course a viable goal in its own right, but they are also directly related to the required energy for a computation.
At this point, we note that while scalability is necessary for 
efficient large scale parallel computing,
scalability alone would not imply an efficient use of resources. 
In fact, inefficient codes are often found to scale better than efficient ones.
Similarly, the asymptotic convergence rate of a discretization scheme is an important
mathematical criterion affecting the accuracy, but ultimately only the error itself matters
including the constants involved.
Such considerations gain additional relevance  at a time when Moore's law
slows down and technological progress will not produce computers anymore
that automatically run twice as fast with every new year.
In this situation, innovation and improvements must rely 
increasingly on better implementations and on algorithms that are 
better suited for the available architectures.

Considering efficiency in this more rigorous sense,
it is found that data transport and not only the executed operations 
are critical factors.
Here, data transport does not only include message passing
communication in a large parallel cluster,
but also the data transport within each node of such a cluster, i.e.,~from main memory to
the CPU, and even within a CPU between the different layers of caches and
the registers of the functional units \cite{hager2010introduction}.
The energy consumption for operations and data transport in a typical CPU architecture
has been quantified in \cite{Afzal2015}.
Additionally, it is of course essential to exploit fine-grained concurrency in the 
form of multi-node architectures and by the use of vectorization.
In order to achieve optimal performance, we must be aware that the current speed of memory 
cannot keep up with the speed of processors and that most of the energy is spent
on the data transfers.
Therefore, an important characteristic relevant 
for the efficiency of numerical algorithms on modern
computers, is their {\em balance} or {\em floating point intensity},
i.e.,~the ratio of floating-point operations (FLOP) 
performed per byte of memory access \cite{hager2010introduction}.

Almost all traditional finite element libraries construct global stiffness
matrices by looping over local elements and adding their contributions to the global matrix.
Even when stored in compressed formats, storing these matrices requires significantly more memory than storing the solution vectors.
Not only the memory consumption presents a challenge, but also the memory traffic and latency in loading the non-zero matrix indices and entries needs to be taken into account.

To improve on the memory consumption and memory access,
matrix-free methods constitute a possible remedy where only the results of matrix vector products are computed without assembling and storing the whole global matrix.
Different strategies exist to implement matrix-free methods,
but the predominant candidate for low-order finite elements is the
element-by-element approach \cite{Arbenz:2008:IJMME,Bielak:2005:CMES,Carey:1986:CANM,Flaig:2012:LSSC,Rietbergen:1996:IJMME}, wherein local stiffness matrices are multiplied by local vectors and later added to the global solution vector.
These local stiffness matrices may either be stored individually in
memory\textemdash{}which actually requires {more} memory than storing the global matrix\textemdash{}or computed on-the-fly.
When using high-order finite elements, the weak forms can be integrated on-the-fly using standard or reduced quadrature formulas \cite{Brown:2010:JSC,Kronbichler:2012:CAF,Ljungkvist:2017:HPC17,Ljungkvist:2017:TechRep,May:2015:CMAME}.
This is a well-suited strategy for future architectures because of its
high arithmetic intensity \cite{loffeld2017arithmetic}\revised{, but we present a method that can compete with matrix-based methods even in the low order case}.
In \cite{scalarstencilscaling}, we presented an alternative matrix-free stencil scaling approach for accelerating low-order finite element implementations 
suited for scalar second-order elliptic partial differential equations (PDE).
There it was shown that the method was able to reduce the computational cost significantly.

We will here expand on this idea and present a similar matrix-free approach for vector-valued \revised{second-order elliptic} PDEs. The construction is based on the use of hierarchical hybrid grids (HHG) which form the basis in the HHG \cite{bergen2004hierarchical, bergen2007hierarchical,
gmeiner2012highly} and HyTeG \cite{Kohl2018HyTeGfinite} frameworks.
\revised{These grids are constructed by starting with an initial, possibly unstructured, simplicial triangulation of a polygonal domain and refining each element multiple times uniformly in order to create a hierarchy of meshes.
Ultimately, we associate to each of these meshes a piecewise linear finite element space.
By exploiting the structure obtained by these uniform refinements, it is possible to improve the performance of the finite element solver by using a stencil based code.}
Vector-valued second-order elliptic PDEs arise in the modeling of elastostatics and fluid dynamics and play an important role in mathematical modeling.
\revised{We show that the idea of the scalar stencil scaling cannot be applied to these equations, since for vector-valued PDEs the simple scaling results in the discretization of a different PDE, if the coefficient is not constant.}
Thus, there is need of a modified stencil scaling method that is also suited for matrix-free finite element implementations on HHGs.
Although this vector-valued scaling is more complicated and more expensive than the scalar stencil scaling, it has the ability to reproduce the standard finite element solutions while requiring only a fraction of the time to obtain them.

The principal novelty is the presentation of an improved method to assemble stencils for vector-valued second-order elliptic PDEs suitable for matrix-free solvers on HHGs \revised{coupled with a linear finite element discretization}.
We provide theoretical and experimental performance comparisons which outline the advantages of the stencil scaling approach.
Furthermore, we show the convergence and the run-times of this extended stencil scaling through numerical experiments.
In these numerical experiments, we consider two specific mathematical models; namely, linear elastostatics, and generalized incompressible Stokes flow.
In the final example, a non-linear shear-thinning non-Newtonian example is considered, where the viscosity depends on the shear rate.

\section{Model equations and discretization}
The goal of this paper is to speed up matrix-free finite element implementations for solving vector-valued second order elliptic PDEs in a domain $\Omega \subset \reals^d$, $d \in \{2,3\}$, of the following form
\begin{equation}
\begin{alignedat}{2}
-\div \stress &= \v{f} &&\quad \text{in } \Omega,\\
\v{u} &= \v{g} &&\quad \text{on } \Gamma_\text{D},\\
\stress \cdot \normal &= \v{\hat{t}} &&\quad \text{on } \Gamma_{\text{N}},
\end{alignedat}
\label{eqn:strongform}
\end{equation}
where the stress $\stress$ = $\stress(\strain)$ depends on the strain and additional material parameters.
One particular example for $\stress$ that we investigate more thoroughly, 
is the stress tensor for linear elasticity with isotropic continuous materials given by Hooke's law as $\stress(\strain) = 2\mu \strain  + \lambda \; \trace{\strain} \identity$.
Furthermore, generalized incompressible Stokes flow problems may also be cast in this form when adding additional constraints.
In this  case, the stress tensor is defined by $\stress(\strain) = 2\mu \strain  - p \identity$, where an additional pressure variable $p$ has been introduced and the incompressibility constraint $\div \v{u} = 0$ in $\Omega$ is enforced.
The domain boundary $\partial\Omega$ is split into two disjoint parts, the \revised{non-trivial} Dirichlet boundary $\Gamma_\text{D}$ and the Neumann boundary $\Gamma_{\text{N}}$.
See \cref{tab:vardefinitions} for a complete list of occurring variables and their definitions.
\begin{table}[h]
\centering
\caption{\label{tab:vardefinitions} Required symbols and their definitions.}
\begin{tabular}{c|l}
\toprule
Symbol & Definition\\ \midrule
$\v{u}$ & displacement or velocity\\
$p$ & pressure\\
$\strain$ & strain: $\frac{1}{2} \left( \grad \v{u} + \grad \v{u}^\top \right)$\\
$\v{f}$ & body forces\\
$\v{g}$ & prescribed displacement or velocity\\
$\v{\hat{t}}$ & external forces\\
$\v{n}$ & outward-pointing unit-normal vector\\
$\lambda$ & Lam\'{e}'s first parameter\\
$\mu$ & shear modulus / dynamic viscosity\\ \bottomrule
\end{tabular}
\end{table}
For the rest of this section, we restrict ourselves to the case of linear elastostatics, since the method may be applied to the momentum balance of the Stokes equations in the same way.
This is demonstrated in the numerical results presented in \Cref{sec:genstokes}.
For simplicity, we consider a homogeneous Dirichlet boundary $\Gamma_{\text{D}}$ \revised{for the rest of this section, so $\v{g} = \v{0}$}.
The weak form of \cref{eqn:strongform} in case of linear elastostatics employing Hooke's law reads: For a suitable space $V$ \revised{incorporating the Dirichlet boundary conditions}, find $\v{u} \in V$ such that $a(\v{u}, \v{v}) = f(\v{v})\; \forall\,\v{v} \in V$, where
\begin{align}
\label{eq:bilinearformstd}
a(\v{u}, \v{v}) &= \scpd[\Omega]{2\mu\strain\left(\v{u}\right)}{\strain\left(\v{v}\right)} + \scpd[\Omega]{\lambda \div \v{u}}{\div \v{v}},\\
f(\v{v}) &= \scpd[\Omega]{\v{f}}{\v{v}} + \scpd[\Gamma_N]{\v{\hat{t}}}{\v{v}}.
\end{align}
By $\scpd[\Omega]{\cdot}{\cdot}$ we denote the standard duality product in $V$.

\revised{In order to discretize the problem, we decompose the computational domain in the typical HHG manner \cite{bergen2005hierarchical,bergen2004hierarchical,bergen2007hierarchical}.
Let $\mesh_{H}$ be a possibly unstructured simplicial triangulation of a bounded polygonal or polyhedral domain 
$\Omega \subset \reals^d$, $d \in \{2,3\}$.
Based on this initial grid, we construct a hierarchy of $L+2$, $L \in \mathbb{N}$, grids $\mesh = \{\mesh_h,\; h = 2^{0}\, H,\ldots,2^{-L-1}\, H\}$ by successive global uniform refinement.
As it is standard, each of these refinements is achieved by subdividing all elements into $2^d$ sub-elements.
For details of the refinement in 3D, we refer to \cite{bey1995tetrahedral}.
We also call $ \mesh_{H}$ macro-triangulation and denote its elements by $T$, whereas the elements of $ \mesh_h$, are denoted by $t_h$.
For better readability, we drop the index $h$ whenever its value is clear from the context.
Since our data structures require at least one interior vertex per macro element, we use the mesh $\mesh_h$ with $h = 2^{-2}H$ as the coarse grid in our multigrid solver hierarchy.
Each of the $\mesh_h$ for $h \leq 2^{-2}H$ has some crucial properties we want to exploit.
The element neighborhood at each vertex in the interior of a macro element is always the same, cf.~\cref{fig:reflected}.
Provided that the coefficient in the PDE is constant, we can construct stencils like in a finite-difference scheme, but with finite elements.
Another important property is that the neighboring elements have some similarities.
In 2D, there are always two elements attached to each stencil edge; see left of \cref{fig:reflected}.
These two elements $t$ and $t^m$ are congruent and they differ only by a reflection along the stencil edge.
A similar structural property holds in 3D which we discuss later in \Cref{sec:corr3d}.

Associated with $ \mesh_h$, is the space $V_h \subset V$ of piecewise linear finite elements.
}\begin{figure}
	\centering
	\begin{tikzpicture}
	\coordinate (one) at (0,0);
	\coordinate (two) at (0,2);
	\coordinate (three1) at (-2,1.5);
	\coordinate (three2) at (2,0.5);
	\coordinate (twoprime) at (0,-2);
	\coordinate (three1prime) at (2,-1.5);
	\coordinate (three2prime) at (-2,-0.5);
	
	\draw[fill=lightgray!60] (one) -- (three1) -- (two);
	\draw[fill=lightgray!60] (one) -- (three2) -- (two);
	\draw[line width=1.5pt] (one) -- (two);
	\draw[line width=1.5pt] (one) -- (three1);
	\draw[line width=1.5pt] (one) -- (three2);
	
	\draw[line width=1.5pt] (one) -- (twoprime);
	\draw[line width=1.5pt] (one) -- (three1prime);
	\draw[line width=1.5pt] (one) -- (three2prime);

	\node[] at ($0.33*(one)+0.33*(two)+0.33*(three1)$) {$t$};
	\node[] at ($0.33*(one)+0.33*(two)+0.33*(three2)$) {$t^m$};
	
	\draw[black,fill=Apricot] (one) circle (5pt);
	\draw[black,fill=black] (two) circle (3pt);
	\draw[black,fill=black] (twoprime) circle (3pt);
	\draw[black,fill=black] (three1) circle (3pt);
	\draw[black,fill=black] (three2) circle (3pt);
	\draw[black,fill=black] (three1prime) circle (3pt);
	\draw[black,fill=black] (three2prime) circle (3pt);
	
	\node[circle,below left=7pt and 4pt,draw,inner sep=2.5pt] at (one) {$i$};
	\node[circle,above right=5pt and 0pt,draw,inner sep=1.5pt] at (two) {$j$};
	
	\end{tikzpicture}
	\hspace*{2cm}
	\includegraphics[width=0.24\textwidth]{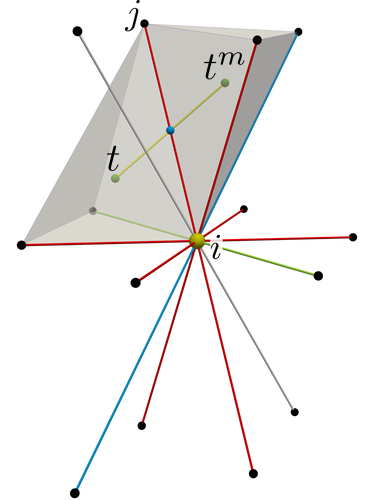}
	\caption{\label{fig:reflected}Element $t$ and reflected element $t^m$ along an edge between nodes $i$ and $j$ in the 2D case (left) and the 3D case (right).}
\end{figure}
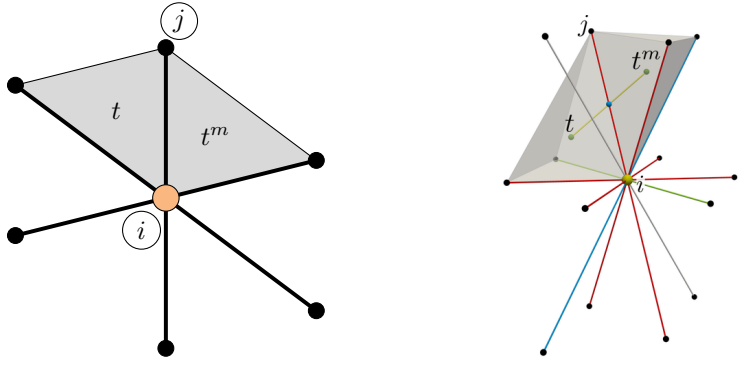
Let $\bfe_i \in \reals^d$ be the canonical unit vector with $(\bfe_i)_j = \delta_{ij}$ for $1 \leq i,j \leq d$, where $\delta_{ij}$ is the Kronecker delta.
Let further $\phi_i \in V_h$ and $\phi_j \in V_h$ be the scalar valued linear nodal basis functions associated to the $i$-th and $j$-th mesh node.
Denote by \revised{$\v{v}_h = \sum_{i} {\v{v}}^{(j)} \phi_{j} $} and \revised{$\v{w}_h = \sum_{j} \v{w}^{(i)} \phi_{i}$} linear combinations of the nodal basis function with vector valued coefficients \revised{${\v{v}}^{(i)} \in \reals^d$ and $\v{w}^{(j)} \in \reals^d$}.
We split the bilinear form \cref{eq:bilinearformstd} in terms of contributions of the bilinear form $a^T$ restricted to each \revised{macro element} $\tet \in \mesh_H$, i.e.,
\begin{equation}
\label{eqn:macrobilinearform}
\begin{aligned}
a(\v{v}_h, \v{w}_h) &= \sum_{\tet \in \mesh_H}^{} a^\tet(\v{v}_h, \v{w}_h) = \sum_{\tet \in \mesh_H}^{} \sum_{i,j} a^\tet(\revised{\v{v}^{(j)}}\phi_j, \revised{\v{w}^{(i)}} \phi_i)\\
&= \sum_{\tet \in \mesh_H}^{} \sum_{i,j} \sum_{l,m=1}^{d} (\revised{\v{v}^{(j)}})_l (\revised{\v{w}^{(i)}})_m \, a^\tet(\phi_j \bfe_l, \phi_i \bfe_m).
\end{aligned}
\end{equation}
In order to simplify notation, we introduce the operator $D$ in place of either differential operator, i.e., $D\v{u} = \epsilon(\v{u})$ or $D\v{u} = \div \v{u}$ and the coefficient placeholder $k$, i.e., $\coeff = \mu$ or $\coeff = \lambda$. For the rest of this subsection, we restrict ourselves to the general bilinear form
\begin{align}
a(\v{u}, \v{v}) &= \scpd[\Omega]{\coeff \cdot D \left(\v{u}\right)}{D \left(\v{v}\right)}.
\end{align}

Using the standard finite element approach, this bilinear form is usually discretized in the following way.
Let $i_t$ and $j_t$ be the local indices of an element $\microtet \in \mesh_h$ associated with the global mesh nodes $i$ and $j$. We denote by $k^t$ the arithmetic mean over all the vertex coefficient values of an element $t$, i.e.,
\begin{align}
\label{eqn:nodalquad}
\bar{\coeff}^t = \frac{\sum_{p=1}^{d+1} \coeff(\bfx^t_p)}{d+1},
\end{align}
where $\bfx^t_p$ are the vertex coordinates of the local element $t$.
Let $\phi^t_i$ and $\phi^t_j$ be the local scalar valued linear nodal basis functions associated to the local vertices $i_t$ and $j_t$.
Because the derivatives of the linear basis functions are constant, employing \cref{eqn:nodalquad} as a quadrature rule to approximate the bilinear form \cref{eqn:macrobilinearform} yields
\begin{align}
\label{eqn:nodalintegration}
a_h(\v{v}_h, \v{w}_h) &=\sum_{\tet \in \mesh_H} \sum_{i,j} \sum_{l,m=1}^{d} (\revised{\v{v}^{(j)}})_l \cdot (\revised{\v{w}^{(i)}})_m \sum_{t \in \mesh_h^{i,j;T}}^{} \bar{\coeff}^t \integral{t}{}{(D (\phi^t_j \bfe_l), D (\phi_i^t \bfe_m))}{x},
\end{align}
where $\mesh_h^{i,j;T}$ is the set of all elements within a \revised{macro element} $T$ adjacent to the edge through $i$ and $j$.
\revised{Note that the number of elements in this set is small due to the HHG structure.
In 2D, there are only two elements adjacent to an interior edge and in 3D there are two types of interior edges.
One type with four elements and a second one with six elements adjacent to it.
Please refer to \Cref{sec:corr3d} for more details.}
Throughout the paper, we denote the bilinear form $a_h(\cdot,\cdot)$ defined in \cref{eqn:nodalintegration} by \emph{nodal integration}.
We note that the choice of $\bar{\coeff}^t$ is natural in case that the coefficient function is stored at the nodes.
Alternative definitions such as the evaluation of the coefficient function at the center of an element are also suitable and preferred if the coefficient function must be evaluated analytically.

With these considerations in mind, we define the reference stencil $\refstencil_{ij}^\tet \in {(\reals^{d \times d})}$ for $\tet \in \mesh_H$ as a $d \times d$ matrix for every pair of mesh nodes $i$ and $j$ by
\begin{align}
\label{eqn:stiffcontribution}
{\left(\refstencil^\tet_{ij}\right)}_{lm} = \integrald{\tet}{}{(D (\phi_j \bfe_l), D (\phi_i \bfe_m))}{x}.
\end{align}
See \cref{fig:reflected} for an illustration of stencils in 2D and 3D in the interior of a \revised{macro element}.
Each edge in the stencil pictures corresponds to a neighbor $j$ of a central entry $i$ in the mesh $\mesh_h$.
\revised{Recall that, as described in the construction of HHGs, the structure in the interior of a single macro element is always the same.
It is also interesting to note that in \cite{scalarstencilscaling} each stencil weight consists of a scalar real value, but here, each stencil weight consists of a $\reals^{d \times d}$ matrix corresponding to the interaction between the dimensional components.}

If $k\equiv 1$, the integrals in \cref{eqn:nodalintegration} may be replaced by the corresponding reference stencils and the bilinear forms \cref{eqn:nodalintegration} and \cref{eqn:macrobilinearform} are equal on the discrete space $V_h \times V_h$.
The following lemma presents a decomposition of the bilinear form \cref{eqn:nodalintegration} which is better suited for matrix-free methods because of its lower operational count while requiring a comparable amount of memory traffic.
This decomposition is very similar to a decomposition of the displacement or velocity field into a symmetric strain rate part and an antisymmetric rotational part.
\begin{lemma}
\label{lemma:rewrite}
Under the assumption that the coefficient $k$ is affine linear on each local element patch $\omega_{i,j;T} = \bigcup_{t \in \mesh_h^{i,j;T}} \overline{t}$, the bilinear form~\cref{eqn:nodalintegration} may be decomposed into a symmetric part with a scaled reference stencil and a remaining antisymmetric correction term $R$:
\begin{align}
\label{eqn:bilinearformscaling}
\hat{a}_h(\v{v}_h, \v{w}_h) &=\sum_{\tet \in \mesh_H} \sum_{i,j} \sum_{l,m=1}^{d} \left(\hat{\coeff}_{ij}^\tet \cdot {(\refstencil^\tet_{ij})}_{lm} + {\left(\corrterm^\tet(\coeff)_{ij}\right)}_{lm}\right) \cdot (\revised{\v{v}^{(j)}})_l (\revised{\v{w}^{(i)}})_m,
\end{align}
where $\hat{\coeff}_{ij}^\tet$ is specified in \cref{eq:patchmean}.
\end{lemma}
\begin{proof}
Let the local stiffness tensor of a local element $t$ be given by
\begin{align}
{\left(\localstiff^t_{ij}\right)}_{lm} = \integral{t}{}{\left(D \left(\phi_j^t \bfe_l\right), D \left(\phi^t_i \bfe_m\right)\right)}{x}.
\end{align}
In the following, we assume that $i \neq j$ and that $k$ is linear on the patch $\omega_{i,j;T}$. Additionally, we introduce the symmetric part $\localstiff^{s;t}_{ij}$ and the antisymmetric part $\localstiff^{a;t}_{ij}$ of $\localstiff^{t}_{ij}$ defined by
\begin{align}
\label{eq:symmmantisymm}
\localstiff^{s;t}_{ij} = \frac{1}{2} \left(\localstiff^{t}_{ij} + \left(\localstiff^{t}_{ij}\right)^\top\right) \;\; \text{and} \;\; \localstiff^{a;t}_{ij} = \frac{1}{2} \left(\localstiff^{t}_{ij} - \left(\localstiff^{t}_{ij}\right)^\top\right).
\end{align}
Due to our mesh structure, for each $t$ in the interior of $T$, there exists a reflected element $t^m$, cf., \cref{fig:reflected}.
Exploiting the fact that $\nabla \phi_i^t = -\nabla \phi_j^{t_m}$, one can show that the local stiffness tensors of these elements are related in the following way:
\begin{align}
\localstiff^{s;t}_{ij} = \localstiff^{s;t^m}_{ij} \;\; \text{and} \;\; \localstiff^{a;t}_{ij} = -\localstiff^{a;t^m}_{ij}.
\end{align}
Before proceeding, we define the following arithmetic mean of the coefficients on the patch $\omega_{i,j;T}$ as follows:
\begin{align}
\label{eq:patchmean}
\hat{\coeff}_{ij}^\tet = \frac{1}{|\mesh_h^{i,j;T}|} \sum_{t \in \mesh_h^{i,j;T}}^{} \bar{k}^t,
\end{align}
where $|\mesh_h^{i,j;T}|$ stands for the number of elements in $\mesh_h^{i,j;T}$.
Using these properties, we can rewrite the last sum in \eqref{eqn:nodalintegration} as
\begin{align}
\sum_{t \in \mesh_h^{i,j;T}}^{} \bar{\coeff}^t {(a^t_{ij})}_{lm} &= \frac{1}{2} \sum_{t \in \mesh_h^{i,j;T}}^{} \bar{\coeff}^t {(a^t_{ij})}_{lm} + \bar{\coeff}^{t^m} {(a^{t^m}_{ij})}_{lm} \nonumber \\
&= \frac{1}{2} \sum_{t \in \mesh_h^{i,j;T}}^{} \bar{\coeff}^t {(\localstiff^{s;t}_{ij})}_{lm} + \bar{\coeff}^t {(\localstiff^{a;t}_{ij})}_{lm} + \bar{\coeff}^{t^m} {(\localstiff^{s;t^m}_{ij})}_{lm} + \bar{\coeff}^{t^m} {(\localstiff^{a;t^m}_{ij})}_{lm}\nonumber \\
&= \frac{1}{2} \sum_{t \in \mesh_h^{i,j;T}}^{} (\bar{\coeff}^t + \bar{\coeff}^{t^m}) {(\localstiff^{s;t}_{ij})}_{lm} + (\bar{\coeff}^t - \bar{\coeff}^{t^m}) {(\localstiff^{a;t}_{ij})}_{lm}\nonumber \\
\label{eqn:stencilscaling1}
&= \hat{\coeff}_{ij}^\tet \cdot \refstencil_{ij} + \frac 12 \sum_{t \in \mesh_h^{i,j;T}}^{} (\bar{\coeff}^t - \bar{\coeff}^{t^m}) {(\localstiff^{a;t}_{ij})}_{lm}.
\end{align}
In the last step, we exploited that for an affine linear $k$, we have $\bar{\coeff}^t + \bar{\coeff}^{t^m} = 2k\left(\frac{x_i + x_j}{2}\right) = 2\hat{\coeff}_{ij}^\tet$ and that
$$\sum_{t \in \mesh_h^{i,j;T}} {(a^{s;t}_{ij})}_{lm} = {(\refstencil^\tet_{ij})}_{lm}.$$
With these considerations in mind, we define the tensor $\corrterm^\tet(\coeff)$ for each $i$ and $j$ by
\begin{align}
\label{eq:corrtermgeneral}
\left(\corrterm^\tet(\coeff)\right)_{ij} = \frac 12 \sum_{t \in \mesh_h^{i,j;T}}^{} (\bar{\coeff}^t - \bar{\coeff}^{t^m}) a^{a;t}_{ij}
\end{align}
In case that $i = j$, we set the correction term $\left(\corrterm^\tet(\coeff)\right)_{ii}$ to zero, the scaling term $\hat{\coeff}_{ii}^\tet$ to $1$, and redefine the central stencil entry as
\begin{align}
\revised{\stencil}^\tet_{ii} = -\sum_{j \neq i} \hat{\coeff}_{ij}^\tet \cdot \refstencil^\tet_{ij} + \left(\corrterm^\tet(\coeff)\right)_{ij}.
\end{align}
This zero-row sum property ensures that translational body motions lie in the kernel of the discrete operator induced by \cref{eqn:bilinearformscaling}.
\end{proof}
In addition to the bilinear form~\cref{eqn:bilinearformscaling}, we define the following form where the correction term $R$ has been omitted:
\begin{align}
\label{eqn:stencilscalingwithoutcorr}
\tilde{a}_h(\v{v}_h, \v{w}_h) &=\sum_{\tet \in \mesh_H} \sum_{i,j} \sum_{l,m=1}^{d} \hat{\coeff}_{ij}^\tet \cdot {(\refstencil^\tet_{ij})}_{lm} (\revised{\v{v}^{(j)}})_l (\revised{\v{w}^{(i)}})_m.
\end{align}
Henceforth, we refer to the bilinear form \cref{eqn:bilinearformscaling} as \emph{physical scaling} and to the form \cref{eqn:stencilscalingwithoutcorr} as \emph{unphysical scaling}.

\revised{
\subsection{Interpretation of the unphysical scaling in 2D}
It may be shown that the bilinear form corresponding to the unphysical form belongs to a different PDE. We illustrate this for the differential operator $ -\div (k\, \strain(\v{u}))$ in 2D.
Particularly, in 2D, a straightforward computation shows  for a differentiable coefficient $k$ and smooth $\v{u}$ that the following identity holds true\begin{align}
    \div (k\, (\div \v{u}) I) = \div (k\, \grad \v{u}^\top) + \begin{pmatrix}
(u_2)_{,y} & -(u_2)_{,x} 
\\
-(u_1)_{,y} & (u_1)_{,x}
\end{pmatrix} \grad k = \div (k\, \grad \v{u}^\top) + (\curl O \v{u}) \grad k,
\end{align}
with the identity matrix $I$ and 
\begin{align}
\curl \v{w} = \begin{pmatrix} (w_1)_{,y} &  -(w_1)_{,x} \\ (w_2)_{,y} &  -(w_2)_{,x}  \end{pmatrix} \text{ and }
O \v{w} = \begin{pmatrix} w_2 \\ - w_1  \end{pmatrix} .
\end{align}
Using the above identity, we find
\begin{align}
 -\div (k\, \strain(\v{u})) & = - \frac 12 \div (k\, \grad \v{u}) -\frac{1}{4} \div (k\, \grad \v{u}^\top) -\frac{1}{4} \div (k\, \grad \v{u}^\top) \\ &=\underbrace{
- \frac 12  \div (k\, \grad \v{u}) - \frac{1}{4} \div (k\, (\div \v{u}) I) -\frac{1}{4} \div (k\, \grad \v{u}^\top)}_{ - {\boldmath \nabla} \cdot \sf{A} (\v{u}) } +  \underbrace{\frac 14 (\curl O \v{u} ) \grad k}_{\sf{B} (\v{u})} \\ &=
- \div \sf{A} (\v{u}) + \sf{B} (\v{u}) .
\end{align}
The first term $\div \sf{A} (\v{u}) $ corresponds to a second order differential operator for which we have
\begin{align}
{\sf{A}} \begin{pmatrix} u \\ 0 \end{pmatrix} : \grad \begin{pmatrix} 0 \\ v \end{pmatrix}= 
{\sf{A}} \begin{pmatrix} 0 \\ u \end{pmatrix} : \grad \begin{pmatrix} v \\ 0 \end{pmatrix} .
\end{align}
This property guarantees  that the antisymmetric part of the associated local stencil term is zero and thus it can be simply scaled.
The second term $\sf{B} (\v{u}) $ is obviously equal to zero if $k$ is constant, otherwise it corresponds to a first order differential operator.
Here we find
\begin{align}
\begin{pmatrix} v & 0 \end{pmatrix} {\sf{B}}\begin{pmatrix} 0 \\ u \end{pmatrix} = - \begin{pmatrix} 0 & v \end{pmatrix} {\sf{B}} \begin{pmatrix} u \\ 0 \end{pmatrix} \;\, \text{and} \;\, \begin{pmatrix} v & 0 \end{pmatrix} {\sf{B}}\begin{pmatrix} u \\ 0 \end{pmatrix} = \begin{pmatrix} 0 & v \end{pmatrix} {\sf{B}} \begin{pmatrix} 0 \\ u \end{pmatrix} = 0,
\end{align}
which yields that the symmetric  part of the associated local stencil term is zero.
 A more detailed comparison shows that it corresponds to the antisymmetric correction term $R$.
Therefore, omitting the correction term $R$ in \cref{eqn:bilinearformscaling} with $D = \strain$ does not result in a discretization of the PDE $-\div (k\, \strain(\v{u})) = \v{f}$ but of $ -\div (\sf{A}(\v{u})) = \v{f}$, for a general $k$.
Only in the case that $k$ is constant in $\Omega$, the correction term vanishes and the original PDE is recovered.
We note that the splitting of the differential operator in the terms associated with the operators $\sf{A} $ and $\sf{B} $ corresponds exactly to the splitting of the
 stencil entries in symmetric and antisymmetric parts. 
Similar considerations can be worked out in 3D.}

\subsection{Stencil-based matrix-free methods}
These newly introduced bilinear forms are very well suited for stencil-based matrix-free methods on HHGs, since the reference stencil and the correction terms are always the same for a single \revised{macro element} and only the scaling terms depending on the coefficient need to be recomputed.
In \Cref{sec:performance}, we present a short analysis of the computational cost of the standard approach by nodal integration compared to the scaling based approaches.

In \cref{lemma:rewrite}, we assume that the coefficient $\coeff$ is affine linear on each local patch of elements adjacent to an edge. Therefore, if $\coeff$ is a global affine linear function, both bilinear forms $a_h(\cdot,\cdot)$ and $\hat{a}_h(\cdot,\cdot)$ are equal.
\revised{Since we use linear finite elements, optimal convergence rates may still be observed if a linear local interpolant of a non-linear $k$ is used.}
The unphysical scaling $\tilde{a}_h(\cdot,\cdot)$, however, is only equal to the other bilinear forms when the coefficient $\coeff$ is constant on the whole domain.

\begin{remark}
The definition in \eqref{eq:patchmean} may be replaced by $\hat{\coeff}_{ij}^\tet = \frac{1}{2}\left(k(\bfx_i) + k(\bfx_j)\right)$.
This approach requires fewer floating point operations but numerical experiments suggest that using \eqref{eq:patchmean} yields better accuracy while not having a huge impact on performance.
The memory traffic for either approach is the same because the coefficients need to be loaded from memory anyway in order to compute the correction term.
\revised{This assumes that the {\em layer condition} \cite{hager2010introduction} is satisfied when traversing the HHG data structures, so the values of $\coeff$ need to be read from memory only once.}
\end{remark}
In practice, this scaling of the reference stencil is only done in the interior of \revised{macro elements} where asymptotically most computations are performed in order to evaluate the bilinear form.
The physically scaled form $\hat{a}_h(\cdot,\cdot)$ is thus redefined as
\begin{equation}
\hat{a}_h(\phi_j \bfe_l,\phi_i \bfe_m)
=
\begin{cases}
a_h(\phi_j \bfe_l,\phi_i \bfe_m)\,,\quad & \text{if } x_i \in \partial T \text{ and } x_j \in \partial T \text{ of at least one } T \in \mesh_H,\\
\hat{a}_h(\phi_j \bfe_l,\phi_i \bfe_m)\,, & \text{otherwise.}
\end{cases}
\end{equation}
This definition enforces global symmetry of the matrix, but requires taking into account special boundary cases when iterating over the interior of \revised{macro elements}.
In practice, we therefore employ an alternative definition where we use the standard bilinear form only if $x_i \in \partial T$ of at least one $T \in \mesh_H$.
This loss of global symmetry across \revised{macro element} interfaces may cause problems for iterative solvers relying on symmetric matrices.
However, this symmetry loss can be regarded as higher order perturbation and in the numerical experiments provided in \Cref{sec:numerical_results}, no degradation of the convergence of the employed iterative solvers could be observed.

In the following two subsections, we show how to efficiently pre-compute most parts of the correction term \eqref{eq:corrtermgeneral} in order to be suitable for stencil based codes.
Since the correction term depends on the space dimension, we derive it separately for 2D and 3D.

\revised{\begin{remark}
The presented idea of scaling the reference stencils and thus obtaining an accurately enough approximation of the stiffness matrix entries is only valid for linear finite elements.
However, for higher-order discretizations it is still possible to exploit the HHG structure.
There, we assume that the reference basis functions, their gradients, and the coefficient are evaluated and stored at the quadrature points.
Using these values, the stencils are assembled similarly as in \cite{Kronbichler:2012:CAF}.
Let $\mathcal{Q}_t$ be the set of quadrature points in an element $t \in \mesh_h$.
Due to the HHG structure, for each $q_t \in \mathcal{Q}_t$ there exists a reflected quadrature point $q_{t^m} \in \mathcal{Q}_{t^m}$.
For each pair of reflected elements $t$ and $t^m$ attached to an edge between two nodes $x_i$ and $x_j$ we need to store the matrices $(E_{q_t})_{lm} = (D (\phi_j(x_{q_t}) \bfe_l), D (\phi_i(x_{q_t}) \bfe_m))$ for $q_t \in \mathcal{Q}_t$.
Let be $\mesh_h^{i,j;T;m}$ the set of half of the elements within a macro element $T$ adjacent to the edge through $i$ and $j$ which have a unique corresponding reflected element which is not in the set.
Additionally, let $\omega_{q_t}$ be the quadrature weights corresponding to the quadrature points in $\mathcal{Q}_t$.
The stencil $\stencil^T_{ij}$ may then by computed via
\begin{align}
\stencil^T_{ij} = \sum_{t \in \mesh_h^{i,j;T;m}} \sum_{q_t \in \mathcal{Q}_t} |t| \omega_{q_t} \left(k_{q_t} + k_{q_{t^m}}\right) E_{q_t}.
\end{align}
\end{remark}}

\subsection{Correction term in 2D}
In this subsection, we consider the correction term \cref{eq:corrtermgeneral} in the case of two dimensions, i.e., $d = 2$, and present a closed form of its values. The antisymmetric part $a^{a;t}_{ij}$ of $a^{t}_{ij}$ \revised{defined through \cref{eq:symmmantisymm} is determined by a single variable $\gamma^{(i,j);t}$ and is of the following form}
\begin{align}
\label{eqn:antisymmetric2d}
\localstiff^{a;t}_{ij} = \begin{pmatrix}
0 & -\gamma^{(i,j);t}\\
\gamma^{(i,j);t} & 0
\end{pmatrix}.
\end{align}
Let \revised{$\v{n}(x) = (n_1(x), n_2(x))^\top$} be the outward pointing unit-normal of an element $t \in \mesh_h$ \revised{for $x \in
\partial t$} and \revised{$\bm{\tau}(x) = (-n_2(x), n_1(x))^\top$} the corresponding tangential vector.
In the following, we assume that the differential operator $D$ is given by $D\v{u} = \strain(\v{u})$.
Additionally, in the constant coefficient reference case, we have $\coeff = 1$.
Doing the same computations with $D\v{u} = \div \v{u}$ results in the same values just with a flipped sign.
The value $\gamma^{(i,j);t}$ can be rewritten as
\begin{align}
\gamma^{(i,j);t} &= \integral{\microtet}{}{\strain\left(\phi_{j} \bfe_1 \right) : \strain\left( \phi_{i} \bfe_2 \right) - \strain\left( \phi_{j} \bfe_2 \right) : \strain\left( \phi_{i} \bfe_1 \right)}{x}
= \frac{1}{2}\integral{\microtet}{}{\phi_{j,y} \phi_{i,x} - \phi_{j,x} \phi_{i,y}}{x}\\
&= \frac{1}{2}\integral{\partial \microtet}{}{\phi_{i} \left(\phi_{j,y} n_1 - \phi_{j,x} n_2\right)}{s}
= \frac{1}{2}\integral{\partial \microtet}{}{\phi_{i} \left(\phi_{j,y} \tau_2 + \phi_{j,x} \tau_1\right)}{s}
= \frac{1}{2}\integral{\partial \microtet}{}{\phi_{i} \grad \phi_{j} \cdot \bm{\tau}}{s}.
\end{align}
\begin{figure}	\centering
	\begin{minipage}{0.33\textwidth}
		\centering
		\begin{tikzpicture}
		\tikzset{font={\fontsize{8pt}{12}\selectfont}}
		\begin{scope}[auto, every node/.style={minimum size=1.6em,inner sep=1},node distance=2cm]
		
		\coordinate (mcc) at (-2,0);
		\coordinate (kc) at (0,0);
		\coordinate (qc) at (-2,2.5);
		\coordinate (pc) at ($.333*(mcc)+.333*(kc)+.333*(qc)$);
		\coordinate (qprimec) at (0,-2.5);
		\coordinate (pprimec) at ($.333*(mcc)+.333*(kc)+.333*(qprimec)$);
		
		\node[draw,circle] at (mcc) (mc) {$i_t$};
		\node[draw,circle] at (kc) (k) {$j_t$};
		\node[] at (pc) (p) {$t$};
		\node[draw,circle] at (qc) (q) {$p_t$};
		\node[] at (pprimec) (tm) {$t^m$};
		\node[draw,circle] at (qprimec) (pprime) {$p_{t^m}$};
		
		\draw[line width=2pt] (k) -- (mc);
		\draw[line width=2pt] (mc) -- (q);
		\draw[line width=2pt] (q) -- (k);
		
		\draw[line width=2pt] (mc) -- (pprime);
		\draw[line width=2pt] (k) -- (pprime);

		\end{scope}
		\end{tikzpicture}
	\end{minipage}
	\hspace*{0.1em}
	\begin{minipage}{0.33\textwidth}
		\centering
		\begin{tikzpicture}
		\tikzset{font={\fontsize{8pt}{12}\selectfont}}
		\begin{scope}[auto, every node/.style={,minimum size=1.6em,inner sep=1},node distance=2cm]
		\coordinate (mcc) at (-2,0);
		\coordinate (kc) at (0,0);
		\coordinate (qc) at (-2,2.5);
		\coordinate (pc) at ($.333*(mcc)+.333*(kc)+.333*(qc)$);
		\coordinate (pprimec) at (-0.5,-1);
		\coordinate (qprimec) at (0,-2.5);
		\node[draw,circle] at (mcc) (mc2) {$i_t$};
		\node[draw,circle] at (kc) (k2) {$j_t$};
		\node[draw,circle] at (qc) (q2) {$p_t$};
		\node[] at (pc) (p) {$t$};
		\draw[line width=1pt,red] ($ 0.5*(k2) + 0.5*(mc2)$) -- +(0.0,-0.75) node[draw=none,below] {$E_1$};
		\draw[line width=1pt,red] ($ 0.5*(k2) + 0.5*(q2)$) -- +(0.5669467,0.42521003213538067) node[draw=none,right] {$E_2$};
		\draw[line width=1pt,red] ($ 0.5*(q2) + 0.5*(mc2)$) -- +(-0.75,0.0) node[draw=none,left] {$E_3$};
		\draw[line width=2pt] (mc2) -- (k2);
		\draw[line width=2pt] (mc2) -- (q2);
		\draw[line width=2pt] (k2) -- (q2);
		\end{scope}
		\end{tikzpicture}
	\end{minipage}
	\caption{\label{fig:localtriangles} Local indices of an element $t$ and its corresponding reflected element $t^m$ (left). An element $t$ with the three edges (right).}
\end{figure}
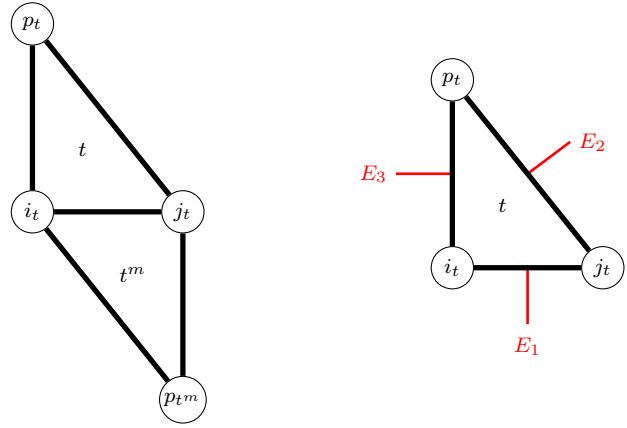
We denote the three edges of an element $t \in \mesh_h$ by $E_1$, $E_2$, and $E_3$ as illustrated in \cref{fig:localtriangles}.
Since $\phi_{i} = 0$ on $E_2$ and $\grad \phi_{j} \cdot \bm{\tau} = 0$ on $E_3$, the integral is reduced to
\begin{align}
\gamma^{(i,j);t} = \frac{1}{2}\integral{\partial \microtet}{}{\phi_{i} \grad \phi_{j} \cdot \bm{\tau}}{s} &= \frac{1}{2}\integral{E_1}{}{\phi_{i} \grad \phi_{j} \cdot \bm{\tau}}{s} = \frac{1}{2 |E_1|}\integral{E_1}{}{\phi_{i}}{s} = \frac{1}{4}.
\end{align}
This constant antisymmetric part $a^{a;t}_{ij}$ needs to be scaled by a difference of coefficients evaluated at the vertices. Using the notation from \cref{fig:localtriangles}, the difference is obtained by
\begin{align}
\bar{\coeff}^t - \bar{\coeff}^{t^m} = \frac{1}{3} \left( \coeff \left(\bfx_{p_t}\right)- \coeff \left(\bfx_{p_{t^m}}\right) \right).
\end{align}
Finally, the correction term evaluates to
\begin{align}
\label{eqn:correction2d}
\left(\corrterm^\tet(\coeff)\right)_{ij} = \frac{1}{12} \left( \coeff \left(\bfx_{p_t}\right)- \coeff \left(\bfx_{p_{t^m}}\right) \right) \begin{pmatrix}
0 & -1\\
1 & 0
\end{pmatrix}.
\end{align}
Note that in the 2D case, the correction term is independent of the geometry, thus no extra information has to be stored in memory.
As can be seen in the next subsection, this is not the case in 3D anymore.

\subsection{Correction term in 3D}
\label{sec:corr3d}
In the 3D case, the uniform grid refinement rule following \cite{bey1995tetrahedral} yields three sub-classes of tetrahedra for each macro element. We denote each of these classes by a color, namely gray, blue, and green, cf.,~left and second from left in \cref{fig:subtypes_and_stencil}. We always associate the class corresponding to the macro element to the gray color. The remaining classes are arbitrarily associated to the colors blue and green.
This uniform refinement results in a stencil which is the same for each interior node of a macro element. The resulting stencil is illustrated in the right of \cref{fig:subtypes_and_stencil}.
We denote the edges adjacent to elements of classes blue and green only, as edges of gray type, cf.,~third from left in \cref{fig:subtypes_and_stencil}.
The edges of green and gray type are defined similarly: An edge of blue type is adjacent to only elements of class green and gray, whereas an edge of green type is adjacent to only elements of class blue and gray.
All other remaining edges are denoted as red-type edges.
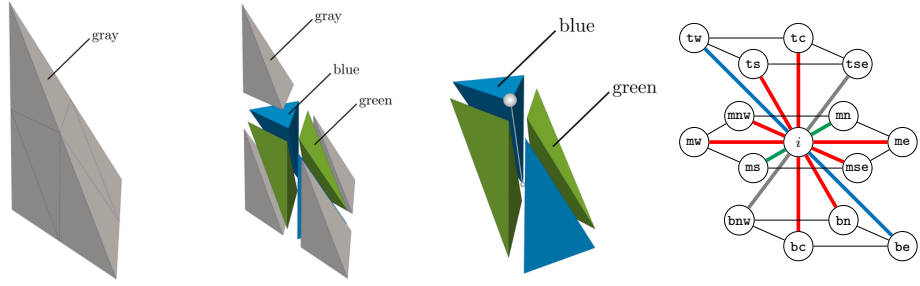
\begin{figure}\centering
\begin{minipage}{0.22\textwidth}
\includegraphics[width=\textwidth]{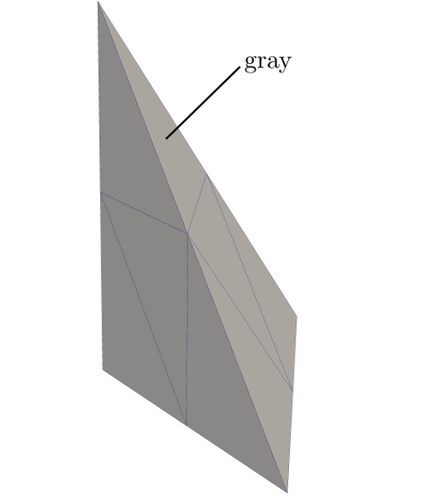}
\end{minipage}\begin{minipage}{0.22\textwidth}
\includegraphics[width=\textwidth]{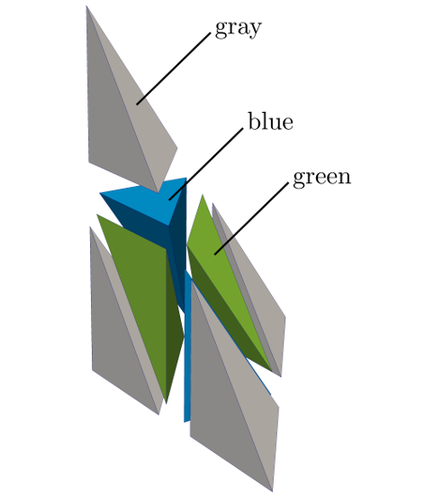}
\end{minipage}\begin{minipage}{0.22\textwidth}
\includegraphics[width=\textwidth]{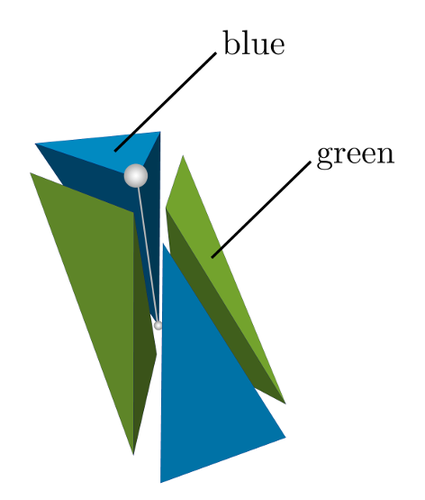}
\end{minipage}\begin{minipage}{0.23\textwidth}
\resizebox {\textwidth} {!} {
\begin{tikzpicture}[x={(0:1cm)}, y={(90:1cm)}, z={(210:0.5cm)}]
\tikzset{font={\fontsize{8pt}{12}\selectfont}}

\begin{scope}[auto, shift={(0,0)}, scale=2,every node/.style={draw,circle,minimum size=1.6em,inner sep=1},node distance=2cm]

\node[draw,circle] at (-1,0,0) (mw) {\texttt{mw}};
\node[draw,circle] at (0,0,0) (mc) {$i$};
\node[draw,circle] at (1,0,0) (me) {\texttt{me}};
\node[draw,circle] at (-1,0,-1) (mnw) {\texttt{mnw}};
\node[draw,circle] at (0,0,-1) (mn) {\texttt{mn}};
\node[draw,circle] at (0,0,1) (ms) {\texttt{ms}};
\node[draw,circle] at (1,0,1) (mse) {\texttt{mse}};

\node[draw,circle] at (0,1,1) (ts) {\texttt{ts}};
\node[draw,circle] at (1,1,1) (tse) {\texttt{tse}};
\node[draw,circle] at (-1,1,0) (tw) {\texttt{tw}};
\node[draw,circle] at (0,1,0) (tc) {\texttt{tc}};

\node[draw,circle] at (0,-1,0) (bc) {\texttt{bc}};
\node[draw,circle] at (1,-1,0) (be) {\texttt{be}};
\node[draw,circle] at (-1,-1,-1) (bnw) {\texttt{bnw}};
\node[draw,circle] at (0,-1,-1) (bn) {\texttt{bn}};

\draw[line width=2pt, red] (mc) -- (me);
\draw[line width=2pt, red] (mc) -- (mnw);
\draw[line width=2pt, ForestGreen] (mc) -- (mn);
\draw[line width=2pt, red] (mc) -- (ts);
\draw[line width=2pt, gray] (mc) -- (tse);
\draw[line width=2pt, RoyalBlue] (mc) -- (tw);
\draw[line width=2pt, red] (mc) -- (tc);
\draw[line width=2pt, red] (mc) -- (bc);
\draw[line width=2pt, RoyalBlue] (mc) -- (be);
\draw[line width=2pt, gray] (mc) -- (bnw);
\draw[line width=2pt, red] (mc) -- (bn);
\draw[line width=2pt, ForestGreen] (mc) -- (ms);
\draw[line width=2pt, red] (mc) -- (mse);
\draw[line width=2pt, red] (mc) -- (mw);

\draw (mnw) -- (mw) -- (ms) -- (mse) -- (me) -- (mn) -- (mnw);
\draw (tc) -- (tw) -- (ts) -- (tse) -- (tc);
\draw (bc) -- (be) -- (bn) -- (bnw) -- (bc);

\end{scope}
\end{tikzpicture}
}
\end{minipage}\caption{\label{fig:subtypes_and_stencil}Uniform refinement of one \revised{macro element} (left) into three sub-classes (second from left). Gray edge adjacent to blue and green sub-tetrahedra (third from left). Stencil at an interior node $i$ of a macro tetrahedron with off-center nodes $j \in \lbrace \ttme, \ttmnw, \ttmn, \ttts, \tttse, \tttw, \tttc, \ttbc, \ttbe, \ttbnw, \ttbn, \ttms, \ttmse, \ttmw \rbrace$ (right).}
\end{figure}In contrast to the 2D case, the general structure of the antisymmetric part of the local stiffness tensor $\localstiff^{a;t}$ \revised{as defined in \cref{eq:symmmantisymm}} for an element $t \in \mesh_h^{i,j;T}$ in 3D is \revised{determined} by three independent values $\gamma$, $\beta$, and $\delta$:
\begin{align}
\label{eqn:antisymmetric3d}
\localstiff^{a;t}_{ij} = \begin{pmatrix}
0 & -\gamma^{(i,j);t} & -\beta^{(i,j);t}\\
\gamma^{(i,j);t} & 0 & -\delta^{(i,j);t}\\
\beta^{(i,j);t} & \delta^{(i,j);t} & 0
\end{pmatrix}.
\end{align}
Let \revised{$\v{n}(x) = (n_1(x), n_2(x), n_3(x))^\top$} be the outward pointing unit-normal of an element $t \in \mesh_h$ \revised{for $x \in
	\partial t$}.
In the case of $D\v{u} = \strain(\v{u})$ and $\coeff = 1$, the non-zero components of~\cref{eqn:antisymmetric3d} evaluate to
\begin{small}
\begin{align}
\beta^{(i,j);t} &= \integral{\microtet}{}{\strain\left(\phi_{j} \bfe_1 \right) : \strain\left( \phi_{i} \bfe_3 \right) - \strain\left( \phi_{j} \bfe_3 \right) : \strain\left( \phi_{i} \bfe_1 \right)}{x}
= \frac{1}{2}\integral{\microtet}{}{\phi_{j,z} \phi_{i,x} - \phi_{j,x} \phi_{i,z}}{x}\\
&= -\frac{1}{2}\integral{\microtet}{}{\phi_{j,xz} \phi_{i} - \phi_{j,xz} \phi_{i}}{x} + \frac{1}{2}\integral{\partial \microtet}{}{\phi_{j,z} \phi_{i} n_1 - \phi_{j,x} \phi_{i} n_3}{s}
= \frac{1}{2}\integral{\partial \microtet}{}{\phi_{i} \left(\phi_{j,z} n_1 - \phi_{j,x} n_3\right)}{s}.
\end{align}
\end{small}
Similarly, the other components may be rewritten in terms of boundary integrals
\begin{align}
\gamma^{(i,j);t} = \frac{1}{2}\integral{\partial \microtet}{}{\phi_{i} \left(\phi_{j,y} n_1 - \phi_{j,x} n_2\right)}{s}\;\;\;\text{and}\;\;\;\delta^{(i,j);t} = \frac{1}{2}\integral{\partial \microtet}{}{\phi_{i} \left(\phi_{j,z} n_2 - \phi_{j,y} n_3\right)}{s}.
\end{align}
Equally as in 2D, in the case that $D\v{u} = \div \v{u}$, the values of all three variables are the same and only the sign is flipped.
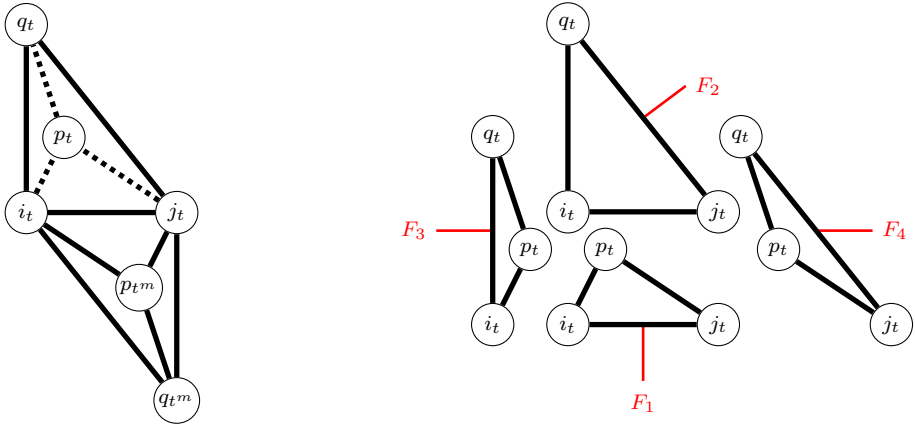
\begin{figure}\centering
\begin{minipage}{0.2\textwidth}
\centering
\begin{tikzpicture}
\tikzset{font={\fontsize{8pt}{12}\selectfont}}
\begin{scope}[auto, every node/.style={draw,circle,minimum size=1.6em,inner sep=1},node distance=2cm]

\coordinate (mcc) at (-2,0);
\coordinate (kc) at (0,0);
\coordinate (pc) at (-1.5,1);
\coordinate (qc) at (-2,2.5);
\coordinate (pprimec) at (-0.5,-1);
\coordinate (qprimec) at (0,-2.5);

\node[draw,circle] at (mcc) (mc) {$i_t$};
\node[draw,circle] at (kc) (k) {$j_t$};
\node[draw,circle] at (pc) (p) {$p_t$};
\node[draw,circle] at (qc) (q) {$q_t$};
\node[draw,circle] at (pprimec) (pprime) {$p_{t^m}$};
\node[draw,circle] at (qprimec) (qprime) {$q_{t^m}$};

\draw[line width=2pt] (k) -- (mc);
\draw[line width=2pt] (mc) -- (q);
\draw[line width=2pt] (q) -- (k);
\draw[line width=2pt, dotted] (mc) -- (p);
\draw[line width=2pt, dotted] (p) -- (q);
\draw[line width=2pt, dotted] (k) -- (p);

\draw[line width=2pt] (mc) -- (pprime);
\draw[line width=2pt] (mc) -- (qprime);
\draw[line width=2pt] (k) -- (pprime);
\draw[line width=2pt] (k) -- (qprime);
\draw[line width=2pt] (pprime) -- (qprime);

\end{scope}
\end{tikzpicture}
\end{minipage}
\hfill
\begin{minipage}{0.77\textwidth}
\centering
\begin{tikzpicture}
\tikzset{font={\fontsize{8pt}{12}\selectfont}}
\begin{scope}[auto, every node/.style={draw,circle,minimum size=1.6em,inner sep=1},node distance=2cm]

\coordinate (mcc) at (-2,0);
\coordinate (kc) at (0,0);
\coordinate (pc) at (-1.5,1);
\coordinate (qc) at (-2,2.5);
\coordinate (pprimec) at (-0.5,-1);
\coordinate (qprimec) at (0,-2.5);

\begin{scope}
\node[draw,circle] at ($(mcc)+(-1, 0)$) (mc3) {$i_t$};
\node[draw,circle] at ($(pc)+(-1, 0)$) (p3) {$p_t$};
\node[draw,circle] at ($(qc)+(-1, 0)$) (q3) {$q_t$};
\draw[line width=1pt,red] ($ 0.5*(mc3) + 0.5*(q3)$) -- +(-0.75,0) node[draw=none,left] {$F_3$};
\draw[line width=2pt] (mc3) -- (q3);
\draw[line width=2pt] (mc3) -- (p3);
\draw[line width=2pt] (p3) -- (q3);
\end{scope}

\begin{scope}
\node[draw,circle] at (mcc) (mc1) {$i_t$};
\node[draw,circle] at (kc) (k1) {$j_t$};
\node[draw,circle] at (pc) (p1) {$p_t$};
\draw[line width=1pt,red] ($ 0.5*(mc1) + 0.5*(k1)$) -- +(0,-0.75) node[draw=none,below] {$F_1$};
\draw[line width=2pt] (mc1) -- (k1);
\draw[line width=2pt] (mc1) -- (p1);
\draw[line width=2pt] (k1) -- (p1);
\end{scope}

\begin{scope}
\node[draw,circle] at ($(kc)+(2.3, 0)$) (k4) {$j_t$};
\node[draw,circle] at ($(pc)+(2.3, 0)$) (p4) {$p_t$};
\node[draw,circle] at ($(qc)+(2.3, 0)$) (q4) {$q_t$};
\draw[line width=1pt,red] ($ 0.5*(q4) + 0.5*(k4)$) -- +(0.75,0) node[draw=none,right] {$F_4$};
\draw[line width=2pt] (k4) -- (p4);
\draw[line width=2pt] (k4) -- (q4);
\draw[line width=2pt] (p4) -- (q4);
\end{scope}

\begin{scope}
\node[draw,circle] at ($(mcc)+(0, 1.5)$) (mc2) {$i_t$};
\node[draw,circle] at ($(kc)+(0, 1.5)$) (k2) {$j_t$};
\node[draw,circle] at ($(qc)+(0, 1.5)$) (q2) {$q_t$};
\draw[line width=1pt,red] ($ 0.5*(k2) + 0.5*(q2)$) -- +(0.5669467,0.42521003213538067) node[draw=none,right] {$F_2$};
\draw[line width=2pt] (mc2) -- (k2);
\draw[line width=2pt] (mc2) -- (q2);
\draw[line width=2pt] (k2) -- (q2);
\end{scope}

\end{scope}
\end{tikzpicture}
\end{minipage}
\caption{\label{fig:localtets} Local indices of an element $t$ and its corresponding reflected element $t^m$ (left). Exploded view of an element $t$ depicting the four faces (right).}
\end{figure}
Let $i_t$, $j_t$, $p_t$, and $q_t$ be the vertex indices of a tetrahedron $\microtet \in \mesh_h$, where $i_t$ and $j_t$ correspond to the global nodes $i$ and $j$; see \cref{fig:localtets}.
Additionally, the corresponding vertex coordinates of these nodes are denoted by an $\bfx$ with a subscript.
The four faces of $\microtet$ are defined by the following triplets of vertices
\begin{align}
F_1 \equiv \{i_t, j_t, p_t\}, \;
F_2 \equiv \{i_t, j_t, q_t\}, \;
F_3 \equiv \{i_t, p_t, q_t\}, \text{ and}\;
F_4 \equiv \{j_t, p_t, q_t\}.
\end{align}
Since $\phi_{i} = 0$ on $F_4$ and $\left(\v{n} \times \grad{\phi_j}\right) = 0$ on $F_3$, applying Stokes' theorem yields
\begin{align}
\begin{pmatrix}
\phantom{-}\delta^{(i,j);t}\\
-\beta^{(i,j);t}\\
\phantom{-}\gamma^{(i,j);t}
\end{pmatrix}&= \frac{1}{2} \integral{\partial \microtet}{}{\phi_i \cdot \left(\v{n} \times \grad{\phi_j}\right)}{s} = \frac{1}{6} \sum_{f=1}^{2} \integral{F_f}{}{\v{n} \times \grad{\phi_j}}{s}\\
&= \frac{1}{6} \sum_{f=1}^{2} \begin{pmatrix}
\integral{F_f}{}{\v{n} \cdot \left( \grad{} \times \revised{\phi_j} \bfe_1\right)}{s}\\
\integral{F_f}{}{\v{n} \cdot \left( \grad{} \times \revised{\phi_j} \bfe_2\right)}{s}\\
\integral{F_f}{}{\v{n} \cdot \left( \grad{} \times \revised{\phi_j} \bfe_3\right)}{s}
\end{pmatrix} = \frac{1}{6} \sum_{f=1}^{2} \begin{pmatrix}
\integral{\partial F_f}{}{\bm{\tau} \cdot \left(\revised{\phi_j} \bfe_1\right)}{s}\\
\integral{\partial F_f}{}{\bm{\tau} \cdot \left(\revised{\phi_j} \bfe_2\right)}{s}\\
\integral{\partial F_f}{}{\bm{\tau} \cdot \left(\revised{\phi_j} \bfe_3\right)}{s}
\end{pmatrix}
= \frac{1}{12} \left(\bfx_{p_t} - \bfx_{q_t}\right),
\end{align}
where $\bm{\tau}$ is the unit tangent. The antisymmetric part of the local stiffness tensor then reduces to
\begin{align}
\label{eqn:localcorr3d}
\localstiff^{a;t}_{ij} = \frac{1}{12}\begin{pmatrix}
0 & (\bfx_{q_t} - \bfx_{p_t})_3 & (\bfx_{p_t} - \bfx_{q_t})_2\\
(\bfx_{p_t} - \bfx_{q_t})_3 & 0 & (\bfx_{q_t} - \bfx_{p_t})_1\\
(\bfx_{q_t} - \bfx_{p_t})_2 & (\bfx_{p_t} - \bfx_{q_t})_1 & 0
\end{pmatrix}.
\end{align}

\begin{figure}\centering
\begin{minipage}{0.3\textwidth}
\centering
\begin{tikzpicture}
\tikzset{font={\fontsize{8pt}{12}\selectfont}}
\begin{scope}[auto, every node/.style={draw,circle,minimum size=1.6em,inner sep=1},node distance=2cm]
\node[draw,circle] at (-2,0) (mw) {$i$};
\node[draw,circle] at (0,0) (mc) {$j$};
\node[draw,circle] at (-1.5,1) (mnw) {$p$};
\node[draw,circle] at (-2,2.5) (tw) {$q$};
\node[draw,circle] at (-0.5,-1) (ms) {$p^m$};
\node[draw,circle] at (0,-2.5) (bc) {$q^m$};

\draw[line width=2pt, red] (mc) -- (mw);
\draw[line width=2pt, gray, dotted] (mw) -- (mnw);
\draw[line width=2pt, gray, dotted] (mnw) -- (mc);
\draw[line width=2pt, gray] (mw) -- (tw);
\draw[line width=2pt, gray, dotted] (tw) -- (mnw);
\draw[line width=2pt, gray] (mc) -- (tw);

\draw[line width=2pt, gray] (mw) -- (ms);
\draw[line width=2pt, gray] (mc) -- (ms);
\draw[line width=2pt, gray] (mw) -- (bc);
\draw[line width=2pt, gray] (ms) to (bc);
\draw[line width=2pt, gray] (mc) -- (bc);
\end{scope}
\end{tikzpicture}
\end{minipage}
\begin{minipage}{0.3\textwidth}
\centering
\begin{tikzpicture}
\tikzset{font={\fontsize{8pt}{12}\selectfont}}
\begin{scope}[auto, every node/.style={draw,circle,minimum size=1.6em,inner sep=1},node distance=2cm]
\node[draw,circle] at (-2,0) (mw) {$i$};
\node[draw,circle] at (0,0) (mc) {$j$};
\node[draw,circle] at (0,-2.5) (bc) {$q^m$};
\node[draw,circle] at (-1.5,-1.5) (bnw) {$r^m$};
\node[draw,circle] at (-2,2.5) (tw) {$q$};
\node[draw,circle] at (-0.5,1.5) (ts) {$r$};

\draw[line width=2pt, red] (mc) -- (mw);
\draw[line width=2pt, ForestGreen] (ts) -- (tw);
\draw[line width=2pt, ForestGreen] (mc) -- (ts);
\draw[line width=2pt, ForestGreen, dotted] (mc) -- (tw);
\draw[line width=2pt, ForestGreen] (mw) -- (ts);
\draw[line width=2pt, ForestGreen] (mw) -- (tw);

\draw[line width=2pt, ForestGreen, dotted] (mc) -- (bnw);
\draw[line width=2pt, ForestGreen] (bc) -- (bnw);
\draw[line width=2pt, ForestGreen] (bnw) -- (mw);
\draw[line width=2pt, ForestGreen] (mw) -- (bc);
\draw[line width=2pt, ForestGreen] (mc) -- (bc);
\end{scope}
\end{tikzpicture}
\end{minipage}
\begin{minipage}{0.3\textwidth}
\centering
\begin{tikzpicture}
\tikzset{font={\fontsize{8pt}{12}\selectfont}}
\begin{scope}[auto, every node/.style={draw,circle,minimum size=1.6em,inner sep=1},node distance=2cm]
\node[draw,circle] at (-2,0) (mw) {$i$};
\node[draw,circle] at (0,0) (mc) {$j$};
\node[draw,circle] at (-1.5,1) (mnw) {$p$};
\node[draw,circle] at (-1.5,-1.5) (bnw) {$r^m$};
\node[draw,circle] at (-0.5,-1) (ms) {$p^m$};
\node[draw,circle] at (-0.5,1.5) (ts) {$r$};

\node[draw opacity=0.0] (dummy_tw) at (-2,2.5) {};
\node[draw opacity=0.0] (dummy_bc) at (0,-2.5) {};

\draw[line width=2pt, RoyalBlue, dotted] (mnw) -- (mc);
\draw[line width=2pt, RoyalBlue] (mw) -- (bnw);
\draw[line width=2pt, RoyalBlue] (mnw) -- (mw);
\draw[line width=2pt, RoyalBlue] (bnw) -- ($(bnw)!0.5!(mc)$);
\draw[line width=2pt, RoyalBlue, dotted] ($(bnw)!0.5!(mc)$) -- (mc);
\draw[line width=2pt, RoyalBlue, dotted] (bnw) -- (mnw);

\draw[line width=2pt, red, dotted] (mc) -- (mw);
\draw[line width=2pt, RoyalBlue] (mc) -- (ts);
\draw[line width=2pt, RoyalBlue] (mc) -- (ms);
\draw[line width=2pt, RoyalBlue] (mw) -- (ts);
\draw[line width=2pt, RoyalBlue] (mw) -- (ms);

\draw[line width=2pt, RoyalBlue] (ms) -- (ts);
\end{scope}
\end{tikzpicture}
\end{minipage}
\caption{\label{fig:rededgetets} Tetrahedra adjacent to an edge of type red with local indexing of the neighboring nodes.}
\end{figure}
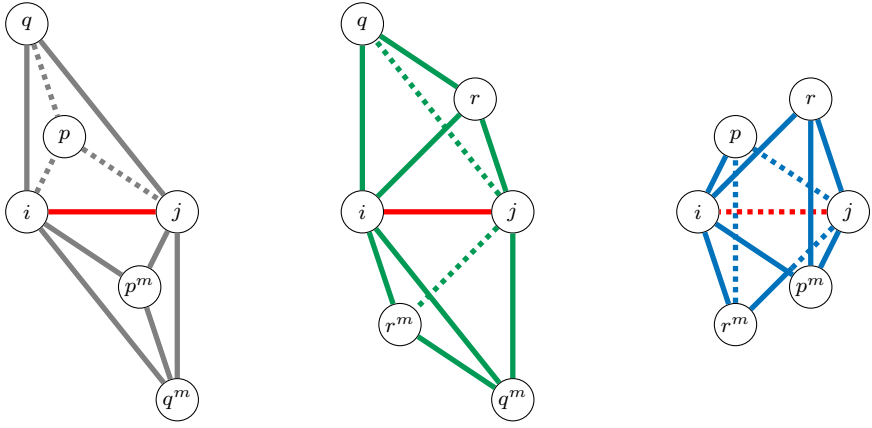
In \cref{fig:rededgetets}, the six elements adjacent to an edge of red type are shown. Each of these tetrahedra belongs to a class which we denote by the colors gray, green, and blue. In each color class, we have eight tetrahedra adjacent to the inner mesh node $i$. We define $\localstiff^{a;t}_{ij}$ to be zero, if elements of the same class as $t$ are not adjacent to an edge which is only the case at blue, gray, and green type edges. We introduce a local indexing of the nodes surrounded by the edge as depicted in \cref{fig:rededgetets} in the following considerations.

Scaling and summing over all elements adjacent to the edge through mesh nodes $i$ and $j$ yields
\begin{small}
\begin{align}
\label{eqn:corr3d1}
\left(\corrterm^\tet(\coeff)\right)_{ij}
= (\bar{k}^{t_{gray}} - \bar{k}^{t_{gray}^m}) a^{a;t_{gray}}_{ij} + (\bar{k}^{t_{green}} - \bar{k}^{t_{green}^m}) a^{a;t_{green}}_{ij} + (\bar{k}^{t_{blue}} - \bar{k}^{t_{blue}^m}) a^{a;t_{blue}}_{ij}.
\end{align}
\end{small}
Since
\begin{align}
\bar{\coeff}^t - \bar{\coeff}^{t^m} = \frac{1}{4}(\coeff(\bfx_{p_t}) + \coeff(\bfx_{q_t}) - \coeff(\bfx_{q_{t^m}}) - \coeff(\bfx_{p_{t^m}})),
\end{align}
we can rewrite~\cref{eqn:corr3d1} by combining terms using the index notation from \cref{fig:rededgetets} as
\begin{align}
\left(\corrterm^\tet(\coeff)\right)_{ij} =& \frac{1}{4}(\coeff(\bfx_{p}) + \coeff(\bfx_{q}) - \coeff(\bfx_{p^m}) - \coeff(\bfx_{q^m})) a^{a;t_{gray}}_{ij}\\
& + \frac{1}{4}(\coeff(\bfx_{q}) + \coeff(\bfx_{r}) - \coeff(\bfx_{q^m}) - \coeff(\bfx_{r^m}))a^{a;t_{green}}_{ij} \\
& + \frac{1}{4}(\coeff(\bfx_{r}) + \coeff(\bfx_{p}) - \coeff(\bfx_{r^m}) - \coeff(\bfx_{p^m})) a^{a;t_{blue}}_{ij}.
\end{align}
After eliminating common sub-expressions, the three additional stencils are defined as
\begin{small}
\begin{align}
\mathcal{S}_{ij}^{T;1} = \frac{1}{4}\left(a_{ij}^{a;t_{gray}} + a_{ij}^{a;t_{blue}}\right),\;\mathcal{S}_{ij}^{T;2} = \frac{1}{4}\left(a_{ij}^{a;t_{gray}} + a_{ij}^{a;t_{green}}\right), \text{ and } \mathcal{S}_{ij}^{T;3} = \frac{1}{4} \left(a_{ij}^{a;t_{green}} + a_{ij}^{a;t_{blue}}\right).
\end{align}
\end{small}
To simplify the notation, we rename the following variables according to \cref{fig:rededgetets} as follows:
\begin{align}
\coeff^{(1)}_{\mathcal{S}^1} = \coeff(\bfx_{p}),\;\coeff^{(2)}_{\mathcal{S}^1} = \coeff(\bfx_{p^m}),\;
\coeff^{(1)}_{\mathcal{S}^2} = \coeff(\bfx_{q}),\;\coeff^{(2)}_{\mathcal{S}^2} = \coeff(\bfx_{q^m}),\;
\coeff^{(1)}_{\mathcal{S}^3} = \coeff(\bfx_{r}), \text{ and } \coeff^{(2)}_{\mathcal{S}^3} = \coeff(\bfx_{r^m}).
\end{align}This yields the following form of the correction term $R^\tet$ in 3D:
\begin{align}
\label{eq:correctionterm3D}
\left(\corrterm^\tet(\coeff)\right)_{ij} = \left(\coeff^{(1)}_{\mathcal{S}^1}-\coeff^{(2)}_{\mathcal{S}^1}\right)\cdot\mathcal{S}^{\tet;1}_{ij} + \left(\coeff^{(1)}_{\mathcal{S}^2}-\coeff^{(2)}_{\mathcal{S}^2}\right)\cdot\mathcal{S}^{\tet;2}_{ij} + \left(\coeff^{(1)}_{\mathcal{S}^3}-\coeff^{(2)}_{\mathcal{S}^3}\right)\cdot\mathcal{S}^{\tet;3}_{ij}.
\end{align}
Ultimately, in addition to the constant reference stencil $\refstencil^T$, we need to store three stencils $\mathcal{S}^{T;1}$, $\mathcal{S}^{T;2}$, and $\mathcal{S}^{T;3}$ per \revised{macro element} in memory. Each of these stencils needs to be scaled appropriately to obtain a computationally cheaper approximation of the bilinear form~\cref{eqn:nodalintegration}.

\section{Mapping piecewise constant coefficients to nodal values}
\label{sec:nodalcoeff}
In many physical applications, the coefficient typically depends on the strain rate $|D(\v{u})|^2$ of the velocity $\v{u}$ and is therefore a constant on each element when using a linear finite element discretization. Since we rely on a stencil based implementation with coefficient values attached to the nodes in $\mesh_h$, we shall discuss efficient techniques to map piecewise constant values to nodal values.
\revised{In this way, index computations that are needed to access the discrete solution can be reused to access the coefficient.
Furthermore, this method is inspired by the ZZ error estimator \cite{zienkiewicz1992superconvergent} where the piecewise constant gradient of a piecewise linear function is lifted to a continuous gradient.
Under certain assumptions, this may improve the accuracy of the coefficient.}
The straight-forward approach would be to find the best approximation of $|D(\v{u})|^2$ in the space of piecewise linear and globally continuous functions with respect to the $L^2$ norm. This, however, involves solving a global linear system and thus is too costly when the coefficient changes after each iteration when solving non-linear problems. Therefore, we refrain from a global $L^2$ projection and focus only on a local technique that is better suited for efficient parallel processing.
One possibility is to assign to each node the volume weighted average of $|D(\v{u})|^2$ over its adjacent elements.
However, we present an alternative method which we also use in our numerical experiments.

In this approach, the discrete function $\v{u}$ is locally projected to an affine linear or quadratic function $\tilde{\v{u}}$ and its derivative is evaluated in order to obtain an approximate value of $|D(\v{u})|^2$ at a node $\bfx_i$.
Let $\mathcal{P}_m(\omega_i)$ be the space of polynomials of order $m$ on the patch $\omega_i$.
The $j$-th component of $\tilde{\v{u}}$ is obtained by solving the following minimization problem
\begin{align}
\label{eq:minimization_problem}
\tilde{u}_j = \argmin_{p \in \mathcal{P}_m(\omega_i)} \sum_{i \in \mathcal{I}_i}\left(p(\bfx_i) - u_j(\bfx_i)\right)^2 \text{ for } m \in \{1,2\},
\end{align}
where $\mathcal{I}_i$ is the index set, containing all indices of the nodal patch $\omega_i$.
Recall that in the case of a uniform refinement in 3D, this involves 15 nodes.
Solving this minimization problem corresponds to solving a small least squares problem for each node $\bfx_i$.
The approximate value of $|D(\v{u})|^2$ evaluated at $\bfx_i$ is then given by $|D(\tilde{\v{u}})(\bfx_i)|^2$,
since by construction $D(\tilde{\v{u}})$ is continuous on $\omega_i$.
As before, the coefficient $\coeff_i$ is obtained in a point-wise fashion according to the physical model.

In the interior of a macro element, the quadratic and affine linear approximations are equivalent.
Particularly, the quadratic minimizing polynomial $p_2$ of \cref{eq:minimization_problem} on a patch $\omega_i$ in the interior of a macro element may be written as $p_2(\bfx) = (\bfx - \bfx_i)^\top \sfA (\bfx - \bfx_i) + \bfb^\top (\bfx - \bfx_i) + c$ for some $\sfA \in \reals^{d\times d}$, $\bfb \in \reals^d$, and $c \in \reals$.
Similarly, a minimizing affine linear function on the same $\omega_i$ may be written as $p_1(\bfx) = \tilde{\bfb}^\top (\bfx - \bfx_i) + \tilde{c}$ for some $\tilde{\bfb} \in \reals^d$, and $\tilde{c} \in \reals$.
Due to the symmetry of the nodes in $\omega_i$, the quadratic and linear parts are decoupled and it follows that $\bfb = \tilde{\bfb}$ and $c = \tilde{c}$.
The derivatives of $p_2$ and $p_1$ are given by $\nabla p_2(\bfx) = \sfA (\bfx - \bfx_i) + \bfb$ and $\nabla p_1(\bfx) = \tilde{\bfb}$.
Evaluating the derivatives at $\bfx_i$, we obtain $\nabla p_2(\bfx_i) = \bfb = \tilde{\bfb} = \nabla p_1(\bfx_i)$.
Therefore, the quadratic approximation is only required on the lower dimensional primitives and the computationally much cheaper affine linear approximation may be used in the interior of macro elements.

\section{Computational cost analysis}
\label{sec:performance}
Since the stencil scaling approach for vector-valued PDEs has been introduced as means to reduce the computational cost for matrix-free finite element implementations, we will present a concise cost analysis.
Asymptotically, most of the computational work is done in the interior of \revised{macro elements}.
Therefore, we restrict our performance analysis to the interior of a single \revised{macro element}.
Furthermore, we ignore all performance impacts stemming from the required communication between processes and focus the analysis on multiple independent processes on a single compute node.
We start with an estimation of the number of required operations to compute the residual $\sfy = \sff - \sfA\sfx$ where the matrix $\sfA$ results from 
a discretization of the vector-valued PDEs with a single scalar coefficient.
Additionally, theoretical estimates on the required memory and the memory traffic are given which are validated by experimental measurements in \Cref{sec:memorytrafficroofline}.
\subsection{Number of operations}
We start by counting the number of operations to compute the residual $\sfy = \sff - \sfA\sfx$ when using either of the presented methods or when storing all the stencils in memory which corresponds to storing the global matrix $\sfA$.
In the case of nodal integration, we assume that the local stiffness matrices (two in 2D and 6 in 3D) are pre-computed and stored in memory.
\revised{In the scaling approach}, we assume that the reference stencils and the additional correction stencils are stored in memory instead.
\revised{We do not pre-compute and store values like the average of $k$ in \cref{eq:patchmean} or the differences of $k$ in \cref{eq:correctionterm3D}.
Since the stencil code is already memory bound, storing only the nodal values of $k$ and computing the averages and differences on-the-fly compared to reading the pre-computed values from memory increases the arithmetic intensity and reduces the required memory and the pressure on the memory buses.}
The required numbers of operations for the different methods are summarized in \cref{tab:operations}.
Please note that these numbers only give estimates on the actual number of instructions performed by the processor since optimizing compilers may reorder, fuse, and vectorize FLOPs, meaning that multiple FLOPs may be performed in a single cycle.
We also ignore the effect of fused multiply-add operations that are typical for most modern CPU architectures.

Let $\sfy_i \in \reals^d$ be the target vector components at position $i$, let $\sfx_i \in \reals^d$ be the input vector components at position $i$, and let $\sff_i \in \reals^d$ be the components of the right-hand side vector at position $i$. 
Furthermore, let $S_{ij} \in \reals^{d \times d}$ be the stencil which acts on a vector at position $j$ in order to obtain the result at position $i$.
The residual for all the degrees of freedom (DoFs) at position $i$ is computed via
\begin{align}
\label{eqn:residuallinalg}
\sfy_i = \sff_i - \sum_{j}^{} S_{ij} \sfx_j.
\end{align}
Assuming that all the $S_{ij}$ for a fixed $i$ are already computed, the number of flops to evaluate \cref{eqn:residuallinalg} is the same for all \revised{three} approaches, as can be seen in the fifth column of \cref{tab:operations}.
In 2D, there are 7 stencils $S_{ij}$ for a fixed $i$, thus 7 local matrix vector multiplications have to be performed and the results are added which results in a total of 26 additions and 28 multiplications. The subtraction from the right-hand side takes 2 extra additions.
Since there are 15 stencils in 3D, similar considerations yield that $15\cdot9 = 135$ multiplications need to be performed.
The number of additions is made up of $15 \cdot 2 \cdot 3 = 90$ additions in the matrix-vector products, $14 \cdot 3 = 42$ additions in the sum over its results, and 3 additions from the subtraction of the right-hand side.

\revised{In the following, we estimate the number of required operations to compute the stencil entries $S_{ij}$ for the matrix-free variants and begin with the physical scaling case.}

\revised{
\subsubsection{Number of operations in the physical scaling case}
Recall that in the unphysical scaling case the stencil is defined as $S_{ij} = \hat{\coeff}_{ij}^\tet \hat{S}_{ij}$ for $j \neq i$.
The central entry for $j = i$ is defined in a way to enforce the zero row sum property, i.e., $S_{ii} = -\sum_{j \neq i}^{} S_{ij}$.

Computing a stencil entry for a fixed $i$ and $j$ with $j \neq i$, requires to compute the value of $\hat{\coeff}_{ij}^\tet$ and scaling the reference stencil.
The calculation of $\hat{\coeff}_{ij}^\tet$ requires 2 multiplications and 3 additions in 2D and in 3D the number of operations depends on the number of elements adjacent to the edge through the nodes $i$ and $j$.
Since we are interested in an upper bound only, we assume the worst case of 2 multiplications and 7 additions.
Finally, due to symmetry, the scaling requires $\frac{d(d+1)}{2}$ multiplications.
These values need to be multiplied by the number of off-center stencils which results in the numbers shown in the third column of \cref{tab:operations}.
Computing the central entry requires 12 additions in each component, totaling in 24 additions in 2D.
In 3D, the number of additions per component is given by 41, resulting in a total of 123 additions.

In order to complete the cost consideration of the physical scaling, the cost of the correction term needs to be assessed.
In 2D, the correction term \cref{eqn:correction2d} just consists of the scaled difference of two coefficient values, which results in a total of 6 subtractions and 6 multiplications. Adding the correction term to the scaled reference stencil requires 12 additional additions.

For the 3D case, recall that the physically scaled stencil is defined as follows: $$S_{ij} = \hat{\coeff}_{ij}^\tet\cdot \hat{S}_{ij} + \left(\coeff^{(1)}_{\mathcal{S}^1}-\coeff^{(2)}_{\mathcal{S}^1}\right)\cdot\mathcal{S}^{\tet;1}_{ij} + \left(\coeff^{(1)}_{\mathcal{S}^2}-\coeff^{(2)}_{\mathcal{S}^2}\right)\cdot\mathcal{S}^{\tet;2}_{ij} + \left(\coeff^{(1)}_{\mathcal{S}^3}-\coeff^{(2)}_{\mathcal{S}^3}\right)\cdot\mathcal{S}^{\tet;3}_{ij}$$ for $j \neq i$.
There we have 3 correction terms with 3 unique non-zero entries for red edges and 2 correction terms with 3 unique non-zero entries for the remaining edges which need to be scaled.
The scaling term of each correction stencil requires one addition only.
This leads to $3 \cdot 3 = 9$ extra multiplications and $3 + 3 \cdot 3 = 12$ additions per red-edge stencil entry.
For the edges of other color $2 \cdot 3 = 6$ extra multiplications and $2 + 2 \cdot 3 = 8$ additions are needed.
Since there are 8 red edges and 6 other edges per stencil, the total number of operations in the third column of \cref{tab:operations} is obtained.}

\subsubsection{Number of operations in the nodal integration case}
\revised{For the number of required operations in the nodal integration case, we recall some of the calculations from the scalar case in \cite{scalarstencilscaling}.
There, the total number of operations is reduced by eliminating common sub-expressions to compute the coefficient value at the quadrature point.
The number of additions required to obtain the non-central
stencil entries in the scalar case is 15 in 2D and 98 in 3D.
In the vector-valued case, almost all these numbers need to be multiplied by 4 in 2D or 9 in 3D, only the sums of the coefficients are computed once per updated node.
The computation of the common sub-expressions in 2D requires 9 additions and 16 in 3D.
Since the common sub-expressions of the coefficient only need to be computed once per node, this results in a total of $4\cdot(15 - 9) + 9 = 33$ in 2D and $9\cdot(98 - 16) + 16 = 754$ in 3D.
The number of operations for the multiplications is obtained by just multiplying the number of scalar operations by 4 or 9, yielding $4 \cdot 12 = 48$ in 2D and $9 \cdot 72 = 648$ in 3D.
In our case, the central entries are not computed by the computationally cheaper method of enforcing the zero row-sum property because the rigid body mode kernel is preserved in this way.
These values are obtained by manually counting the number of operations for the central entries, yielding the values presented in the third and fourth row of \cref{tab:operations}.}

\subsubsection{Number of operations in the stored stencils case}
In this scenario, the whole global matrix $\sfA$ is stored in memory.
Therefore, we assume that no costs are involved in computing the stencil entries and only the operations to compute the residual are required.
Note that this scenario is the preferred one with respect to the number of operations but it consumes the most memory and it has the largest impact on memory traffic from main memory, cf.~\Cref{sec:memorytheory}.

\subsubsection{Comparison of total required operations}
The theoretical analysis of the required operations 
yields estimates of how much CPU time could be saved 
in case that the memory bandwidth is not limited and that
the overhead stemming from index calculations is ignored.
As can be seen, the savings in FLOPs are minor in 2D, but in 3D
they are quite significant.
For 2D, \cref{tab:operations} shows that the physical scaling requires 82\% of the FLOPs that are needed by the on-the-fly nodal integration. In 3D, the physical scaling requires 41\% of the FLOPs needed by the nodal integration.
However, as can be seen in the measurements in \Cref{sec:memorytrafficroofline}, the compiler reduces the number of theoretically estimated FLOPs.
\revised{Using the values reported by the Intel Advisor, the physical scaling requires 44\% of the FLOPs needed by the nodal integration.}

\begin{table}\footnotesize
	\centering
	\caption{\label{tab:operations}Operation count for residual computation.}
	\begin{tabular}{c|c|c|c|c|c}
		\toprule
		method & dimension & non-central entries & central entry & residual & total\\
		\midrule
						\multirow{2}{*}{physical scaling} & 2D & 36 add / 36 mul & 24 add / 0 mul & 28 add / 28 mul & 152\\
		& 3D & 242 add / 220 mul & 123 add / 0 mul & 135 add / 135 mul & 855\\ \midrule
		\multirow{2}{*}{nodal integration} & 2D & 33 add / 48 mul & 24 add / 24 mul & 28 add / 28 mul & 185\\
		& 3D & 754 add / 648 mul & 207 add / 216 mul & 135 add / 135 mul & 2095\\ \midrule
		\multirow{2}{*}{stored stencils} & 2D & 0 add / 0 mul & 0 add / 0 mul & 28 add / 28 mul & 56\\
		& 3D & 0 add / 0 mul & 0 add / 0 mul & 135 add / 135 mul & 270\\
		\bottomrule
	\end{tabular}
\end{table}

\subsection{Memory consumption and memory access}
\label{sec:memorytheory}
For the best performance it is not only required that the number of FLOPs is small, but also the memory traffic from main memory has to be small relative to the required FLOPs.
Therefore, we first give a short summary on the required number of double precision variables that are needed for a residual computation in the interior of a single \revised{macro element} in \cref{tab:memorytotal}, where $N$ is the number of scalar degrees of freedom in the interior of a single \revised{macro element}.
The third column summarizes the number of variables required to store the discretized functions $\sff$, $\sfx$, $\sfy$, and $\sfk$.
The fourth column summarizes the number of variables required to store the discretized operator $\sfA$.
Note that only for the stored stencils approach the memory required to store the operator grows with the mesh size.
The total memory footprint is worst for the stored stencils approach.
In this scenario 135 extra scalar variables must be stored, a number that
would alternatively permit an extra level of refinement of the mesh when using
one of the matrix-free approaches.
Even if storing all stencils is cheapest in terms of FLOPs, it
creates a severe restriction
on the size of the problems that can be solved and
it leads to a very large
amount of data that must be transferred \revised{{from main memory}} in each matrix-vector product.

In \Cref{tab:memoryperdof}, we present estimates on the average number of bytes which need to be loaded from and stored in main memory to compute the residual at a single mesh node in 3D.
\revised{We split the estimation into two extreme cases.}
In the optimistic case, we assume perfect caching and that all previously loaded values stay in the fast cache levels.
In the pessimistic case, we assume no caching at all and that all the data has to be loaded from the slow main memory.
This analysis gives lower and upper bounds on the required main memory traffic and the value observed in practice will lie somewhere between these bounds.
Note that stores and loads of temporary variables required for the computation of the stencil weights are not considered in these estimates.
\revised{These values only present estimates for the number of bytes that must be transferred from main memory,
but they are not necessarily proportional to the time required to load and store them.
In modern architectures, data is moved in terms of cache lines
that may for example be 64 byte large.
If numerical data is stored contiguously, successive values can be accessed more efficiently from cache
lines that are already loaded.
Furthermore, modern CPU microarchitectures employ prefetching that can accelerate the
access to regularly strided data.
A detailed analysis of such effects on the speed of numerical kernels, as
presented in, e.g.,~\cite{alappat2020understanding}, is beyond the scope of this article.
In \cref{tab:memoryperdof}, we present estimated values for the bytes to be transferred in
an optimistic and a pessimistic scenario.}

For the matrix-free variants, the pre-computed stencil values or local stiffness matrices need to be loaded.
\revised{In the physical scaling case, the 15 reference stencils weights for 9 block operators are required which results in a total of 1080 bytes.
Additionally, the 3 additional correction stencils need to be loaded, resulting in 4320 bytes.}
In the nodal integration case, 6 local stiffness matrices with 16 entries each need to be loaded for each of the 9 operators, resulting in 6912 bytes.
In the optimistic case, these data stays in the caches and is loaded from main memory only in the pessimistic case.

Only one coefficient has to be loaded in the optimistic case, 
but in the worst case all 15 coefficients adjacent to a mesh node need to be loaded from main memory which results in 120 bytes per mesh node.
In the stored stencils approach, for each mesh node all the 15 stencil weights for 9 operators need to be loaded even in the optimistic case.

Additionally, the variables $\sff$, $\sfx$, and $\sfy$ are accessed during an iteration.
In the optimistic and pessimistic cases, $24$ bytes of $\sff$ need to be loaded from main memory.
Additionally, because of write allocation, $24$ bytes from $\sfy$ need to be loaded before they are stored, resulting in traffic of $48$ bytes.
Re-using cached values of $\sfx$ in the optimistic case requires loading $24$ bytes, but in the pessimistic case all $15$ neighboring values need to be loaded, resulting in $360$ bytes.

These estimates show that with poor cache re-use
the matrix-free approaches must be expected
to produce even more main memory traffic than the stored stencils approach.
However, when the layer condition is satisfied for the data traversal
and the caches are used efficiently, the matrix-free methods may lead to a reduced main memory traffic.
\begin{table}
\centering
\caption{\label{tab:memorytotal}Number of double precision variables required on a single \revised{macro element} with $N$ scalar degrees of freedom for a residual computation.}
\begin{tabular}{c|c|c|c}
\toprule
method & dimension & variables ($\sff$, $\sfx$, $\sfy$, $\sfk$) & operators \\
\midrule
\multirow{2}{*}{physical scaling} & 2D & $7 \cdot N$ & 28 \\
& 3D & $10 \cdot N$ & 540 \\ \midrule
\multirow{2}{*}{nodal integration} & 2D & $7 \cdot N$ & 72 \\
& 3D & $10 \cdot N$ & 864 \\ \midrule
\multirow{2}{*}{stored stencils} & 2D & $7 \cdot N$ & $28 \cdot N$ \\
& 3D & $10 \cdot N$ & $135 \cdot N$ \\
\bottomrule
\end{tabular}
\end{table}

\begin{table}
\centering
\caption{\label{tab:memoryperdof}Average number of bytes required to load from and store to main memory when computing the residual at a mesh node in 3D assuming the usage of 64-bit double precision floating point variables.}
\begin{tabular}{c|c|c}
	\toprule
	method & optimistic & pessimistic \\
	\midrule
	physical scaling & $8 + 96 = 104$ & $4320 + 120 + 432 = 4872$ \\
    nodal integration & $8 + 96 = 104$ & $6912 + 120 + 432 = 7464$ \\
	stored stencils & $1080 + 96 = 1176$ & $1080 + 432 = 1512$ \\
	\bottomrule
\end{tabular}
\end{table}

\section{Numerical results and applications}
\label{sec:numerical_results}
In this section, we provide numerical results to illustrate the accuracy and run-time of the new scaling approaches in comparison to the assembly by nodal integration in a matrix-free framework.
\revised{Throughout this section, we denote the time-to-solution by
tts and by the relative tts, we denote the ratio of the time-to-solution of the stencil-scaling approach
with respect to the nodal integration.
The numerical solutions obtained by the corresponding bilinear forms $a_h(\cdot,\cdot)$ and $\hat{a}_h(\cdot,\cdot)$ are always denoted by $\v{u}_h$ and $\hat{\v{u}}_h$, respectively.}

\revised{We use two machines to obtain the run-time measurements presented in the following subsections.
Most of the measurements are conducted on the newer SuperMUC-NG system equipped with Skylake nodes.
The following values were taken from \cite{SuperMUCNG:Configuration:Details}.
Each node has two Intel Xeon Platinum 8174 processors with a nominal clock rate of 3.1 GHz.
Each processor has 24 physical cores which results in 48 cores per node.
Each core has a dedicated L1 (data) cache of size 32\,kB and a dedicated L2 cache of size 1024\,kB.
Each of the two processors has a L3 cache of size 33\,MB shared across all its cores.
The total main memory of 94\,GB is split into equal parts across two NUMA domains with one processor each.
We use the Intel 19.0 compiler together with the Intel 2019 MPI library and specify the compiler flags \texttt{-O3 -march=native -xHost}.

The second machine we use for some measurements, is the older SuperMUC Phase~2 system equipped with Haswell nodes.
The following values were taken from \cite{SuperMUC:Configuration:Details}.
Each node has two Intel\textsuperscript{\textregistered} Xeon\textsuperscript{\textregistered} E5-2697 v3 processors with a nominal clock rate of 2.6 GHz.
Each processor has 14 physical cores which results in 28 cores per node.
Each core has a dedicated L1 (data) cache of size 32\,kB and a dedicated L2 cache of size 256\,kB.
The theoretical bandwidths are 343\,GB/s and 92\,GB/s, respectively.
The CPUs are running in cluster-on-die mode.
Thus, each node represents four NUMA domains each consisting of 7 cores with a separate L3 cache of size 18\,MB and a theoretical bandwidth of 39\,GB/s.
On top of this, each NUMA domain has 16 GB of main memory with a theoretical bandwidth of 6.7\,GB/s available.
On this second machine, we use the Intel 18.0 compiler together with the Intel 2018 MPI library and specify the compiler flags \texttt{-O3 -march=native -xHost}.
Note that the serial runs using only a single compute core are not limited to run on large machines like SuperMUC but can also be run on usual modern desktop workstations with enough memory.

All the following experiments were implemented in the HHG framework \cite{bergen2005hierarchical,bergen2004hierarchical,bergen2007hierarchical}.
If not otherwise specified, we solve the linear systems by applying geometric multigrid V-cycles directly to the system until a specified relative tolerance of the norm of the residual is obtained.
The transfers from a coarser to a finer grid are performed by a matrix-free linear interpolation and the restriction is performed by the corresponding transposed matrix-free operation.
As a smoother, we employ the hybrid Gauss-Seidel method, meaning that in the interior of macro elements standard Gauss-Seidel iterations are performed. Across the interfaces not all dependencies are updated which results in a Jacobi-like method.
On the coarse grid, we perform iterations of the diagonally preconditioned conjugate gradient method up to a fixed large relative tolerance or for a fixed number of iterations in order to avoid unnecessary over solving.}

\subsection{Linear elastostatics}\revised{In this subsection, we first perform benchmarks to verify the performance models from \Cref{sec:performance}, and we consider two problems in linear elasticity.} The first \revised{problem} is a benchmark problem where we have an analytical solution at hand and can compute the discretization errors directly. In the second one, a more relevant problem is investigated, where an external force is applied to a metal foam.
\subsubsection{Memory traffic and roofline analysis}
\label{sec:memorytrafficroofline}
In \Cref{sec:memorytheory}, we presented theoretical estimates on the number of floating point operations and the memory accesses required to compute the residual of a linear system using different strategies to obtain the matrix entries.
In this subsection, we verify these results experimentally using a specially designed benchmark, executed on a single compute node of \revised{SuperMUC-NG}.
With this benchmark, we compare the performance of the in \Cref{sec:memorytheory} \revised{analyzed} methods, i.e., the physical stencil scaling, standard nodal integration, and the stored stencils approach.
\revised{The floating-point performance measurements were conducted using the Intel Advisor 2019 \cite{intel-advisor} and the memory traffic measurements were conducted by accessing the hardware performance counters using the Intel VTune Amplifier 2019 \cite{intel-vtune}.}

The benchmark computes the residual $y = f-Ax$, for a vector-valued operator $A$ \revised{in 3D with the same sparsity pattern as the discretized linear elasticity operator}.
As in the theoretical analysis, we only consider the DoFs in the interior of macro tetrahedra.
The residual computation is iterated 500 times in order to obtain an averaged value reducing errors stemming from small fluctuations in the run-time.
The benchmark is executed using \revised{48} MPI \revised{ranks}, pinned to the \revised{48} physical cores of a single node. 
This is essential to avoid optimistic bandwidth values when only a single core accesses the memory.
Measurements with the Intel Advisor \revised{and VTune Amplifier} are carried out solely on rank 0.
Moreover, all measurements are restricted to the inner-most update loop, i.e., where the actual nodal updates take place.
This does not influence the results since the outer loops are identical in all variants.
\revised{Note that in each update, the DoFs corresponding to a single mesh node are updated at once, i.e.,~three DoFs per update.}
We choose $L=5$ as refinement level, which yields $1.09\cdot10^6$ DoFs per \revised{macro element}.
The computation involves three vector-valued variables $x$, $y$, and, $f$ where each of them requires about \SI{8.4}{\mebi\byte} of storage per macro tetrahedron.
In addition to this, the scalar valued coefficient $k$ requires about \SI{2.8}{\mebi\byte} of storage.

We assign \revised{three macro elements} to each MPI \revised{rank} which is the maximum possible for the stored stencils approach \revised{on SuperMUC-NG}.
In practice, the memory limit of a compute node would be reached even faster, since all the lower-dimensional primitives, the multigrid hierarchy, and the communication buffers require extra memory.
Using these settings, each inner-most loop is executed \revised{$500\,062\,500$} times per MPI \revised{rank}.
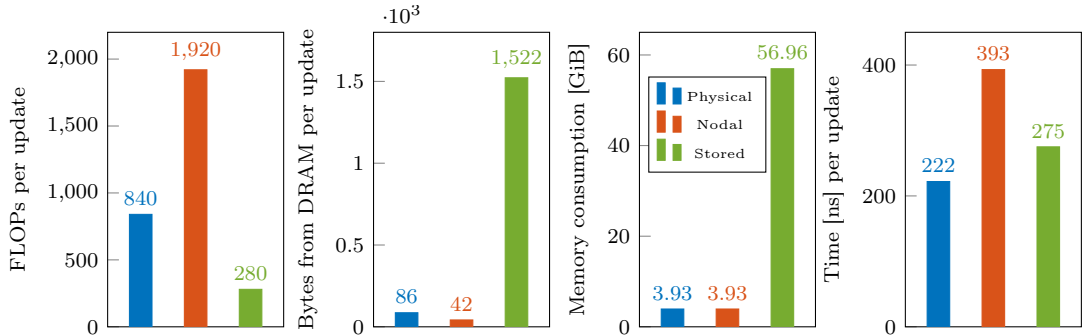
\begin{figure}
	\centering
	\begin{tikzpicture}
	\begin{groupplot}[
	group style={
		group name=performance,
		group size=4 by 1,
		horizontal sep=34pt,
							},
	tickpos=left,
	ytick align=inside,
	xtick align=inside
	]
	\nextgroupplot[font=\footnotesize,
	ymin=0,
	height=5.5cm,
	x=3.6cm,
	enlarge x limits={abs=0.45cm},
	bar width=0.3cm,
	ybar,
	xmajorticks=false,
	nodes near coords,
	xmin=0.0,xmax=0.4,
	ylabel={FLOPs per update},
	ymax=2200]
	\addplot+[color1] coordinates {(0.1,840)};
	\addplot+[color2] coordinates {(0.2,1920)};
	\addplot+[color5] coordinates {(0.3,280)};
	\nextgroupplot[font=\footnotesize,
	ymin=0,
	height=5.5cm,
	x=3.6cm,
	enlarge x limits={abs=0.45cm},
	bar width=0.3cm,
	ybar,
	xmajorticks=false,
	nodes near coords,
	xmin=0.0,xmax=0.4,
	ylabel={Bytes from DRAM per update},
	legend style={anchor=north west},
	ymax=1800,
	scaled y ticks={base 10:-3},
	]
	\addplot+[color1] coordinates {(0.1,86)};
	\addplot+[color2] coordinates {(0.2,42)};
	\addplot+[color5] coordinates {(0.3,1522)};
	\nextgroupplot[font=\footnotesize,
	legend style={font=\tiny},
	ymin=0,
	height=5.5cm,
	x=3.6cm,
	enlarge x limits={abs=0.45cm},
	bar width=0.3cm,
	ybar,
	xmajorticks=false,
	nodes near coords,
	xmin=0.0,xmax=0.4,
	ylabel={Memory consumption [\si{\gibi\byte}]},
	legend style={at={(0.05,0.85)},anchor=north west},
	ymax=65]
	\addplot+[color1] coordinates {(0.1,3.9292368)};
	\addlegendentry{Physical};
	\addplot+[color2] coordinates {(0.2,3.9292368)};
	\addlegendentry{Nodal};
	\addplot+[color5] coordinates {(0.3,56.96064)};
	\addlegendentry{Stored};
	\nextgroupplot[font=\footnotesize,
	ymin=0,
	height=5.5cm,
	x=3.6cm,
	enlarge x limits={abs=0.45cm},
	bar width=0.3cm,
	ybar,
	xmajorticks=false,
	nodes near coords,
	xmin=0.0,xmax=0.4,
	ylabel={Time [\si{\nano\second}] per update},
	ymax=450]
	\addplot+[color1] coordinates {(0.1,222)};
	\addplot+[color2] coordinates {(0.2,393)};
	\addplot+[color5] coordinates {(0.3,275)};
	\end{groupplot}
	\end{tikzpicture}
	\caption{\label{fig:performancebars}\revised{Bar plots for comparing the performance and memory consumption of the nodal integration, physical stencil scaling, and stored stencils approaches on SuperMUC-NG. Three macro tetrahedra are assigned to each MPI rank.}}
\end{figure}

\begin{figure}
\centering
\begin{tikzpicture}
\begin{groupplot}[
group style={
	group name=performance,
	group size=2 by 1,
	horizontal sep=50pt,
			},
tickpos=left,
ytick align=inside,
xtick align=inside
]
\nextgroupplot[font=\footnotesize,
ymin=0,
height=5.5cm,
x=3.6cm,
enlarge x limits={abs=0.45cm},
bar width=0.3cm,
ybar,
xmajorticks=false,
nodes near coords,
xmin=0.0,xmax=0.4,
ylabel={Time [\si{\nano\second}] per update},
ymax=800,
title={SuperMUC Phase 2}]
\addplot+[color1] coordinates {(0.1,248)};
\addplot+[color2] coordinates {(0.2,717)};
\addplot+[color5] coordinates {(0.3,265)};
\nextgroupplot[font=\footnotesize,
ymin=0,
height=5.5cm,
x=3.6cm,
enlarge x limits={abs=0.45cm},
bar width=0.3cm,
ybar,
xmajorticks=false,
nodes near coords={\pgfmathfloatifflags{\pgfplotspointmeta}{0}{}{\pgfmathprintnumber{\pgfplotspointmeta}}},
xmin=0.0,xmax=0.4,
ylabel={Time [\si{\nano\second}] per update},
ymax=800,
title={SuperMUC-NG},
legend style={at={(1.05,0.85)},anchor=north west}]
\addplot+[color1] coordinates {(0.1,222)};
\addlegendentry{Physical};
\addplot+[color2] coordinates {(0.2,393)};
\addlegendentry{Nodal};
\addplot+[color5] coordinates {(0.3,0)};
\node[color=color5,rotate=90] at (axis cs: 0.4,390) {out of memory};
\addlegendentry{Stored};
\end{groupplot}
\end{tikzpicture}
\caption{\label{fig:performancecomparison}\revised{Time per update on different machines with four macro tetrahedra attached to each MPI rank. Left: SuperMUC Phase 2. Right: SuperMUC-NG.}}
\end{figure}
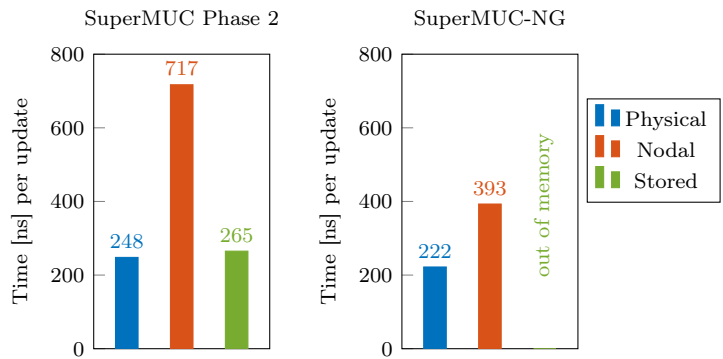

In \cref{fig:performancebars}, we summarize the recorded performance results of the three approaches.
In the leftmost plot of \cref{fig:performancebars}, the FLOPs per update are shown which are close to the theoretically estimated values from~\cref{tab:operations}.
The second from left plot presents the total number of transferred bytes \revised{from main memory} per update.
The total memory consumption is shown in the third from left plot.
The stored stencils approach requires almost 15 times more memory than the matrix-free approaches.
\revised{At first sight, the stored stencils approach looks the most attractive with respect to the required operations when enough main memory is available.
However, the rightmost plot shows that the physical scaling approach has a slightly lower time per update than the stored stencils approach.
This is due to the large amount of data which needs to be transferred from main memory in each update, cf.~second from left plot in \cref{fig:performancebars}.
This means that in fact the caches are more efficiently used in the matrix-free approaches.
Note that these measured values are also close to the theoretically estimated values reported in \cref{tab:memoryperdof}.
}

\revised{The same benchmark was conducted on SuperMUC Phase 2, but with four macro tetrahedra assigned to each of the 28 MPI ranks.
This was possible, since this machine has more memory available per core than SuperMUC-NG.
In \cref{fig:performancecomparison}, we contrast the time per update on SuperMUC Phase 2 with the time per update on SuperMUC-NG using equal problem sizes per MPI rank.
Note that the matrix-free methods both worked on SuperMUC-NG, but the stored stencil approach required too much memory.
Furthermore, the time per update of the physical scaling did not improve much, but the time per update of the nodal on-the-fly integration was reduced by about $45$\%.
This is due to the increased clock rate of SuperMUC-NG compared to SuperMUC Phase 2 and the larger arithmetic intensity of the on-the-fly integration.
}

In order to \revised{further visualize these results}, we present a roofline analysis in \cref{fig:roofline} \revised{conducted on SuperMUC-NG}; see \cite{Ilic:2013:CAL,Williams2009}.
The abscissa shows the arithmetic intensity, i.e.,~the number of FLOPs divided by the number of bytes loaded and stored in the inner-most loop. 
The ordinate gives the measured performance as FLOPs performed per second.
For reference, we added measured saturated memory bandwidth rooflines as reported by the Intel Advisor.
\revised{Obviously, these measured values are smaller than the theoretically
optimal ones given in the hardware specifications.
The maximum performance for double precision vectorized fused multiply-add operations is reported by the Intel Advisor tool as 72.39\,GFLOPs/s.
From the roofline analysis one can see that the nodal integration yields the best performance with respect to FLOPs per second, but is still slower in practice since it requires more than twice the number of operations compared to the physical scaling.
The physical scaling has a smaller arithmetic intensity and a slightly worse performance in GFLOPs/s while the stored stencils approach has the lowest arithmetic intensity with the worst performance.
The compiler could not auto-vectorize the physical scaling and stored stencils kernels.
In the nodal integration kernel, however, one of the fixed-length inner loops could be auto-vectorized.
This explains why the performance of the nodal integration is slightly better than the double precision scalar add performance.}
Of course, the roofline analysis 
constitutes only a first quantitative evaluation
of the performance. Other performance models, like the execution-cache-memory performance model \cite{stengel2015quantifying}, 
can give deeper insight.
\begin{remark}
\revised{As can be seen in \cref{fig:roofline}, the physical scaling approach reaches about $5.25$\% of the peak performance based on double precision vector fused multiply add instructions.
However, considering other rooflines, the physical scaling almost reaches the double precision scalar add performance and reaches about $10.43$\% of the double precision vector add performance.
Further performance optimizations are difficult because of the not ideal mix of multiplies and adds and the challenging vectorization due to the index calculations in tetrahedral elements.
These investigations and performance optimizations are beyond the scope of this paper but are part of the future work and the ongoing development of the software structures in HyTeG \cite{Kohl2018HyTeGfinite}.}
\end{remark}
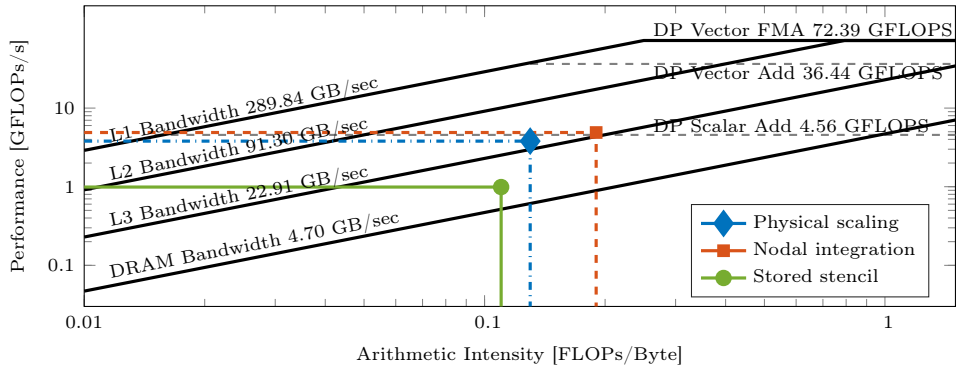
\begin{figure}
	\centering
		\begin{tikzpicture}[font=\scriptsize]
\begin{axis}[width=0.8\textwidth,
height=4cm,
at={(0.769in,0.576in)},
scale only axis,
xmode=log,
xmin=0.01,
xmax=1.5,
xtick={0.01,0.1,1},
xticklabels={{0.01},{ 0.1},{   1}},
xminorticks=true,
xlabel style={font=\color{white!15!black},font=\scriptsize},
xlabel={Arithmetic Intensity [FLOPs/Byte]},
ymode=log,
ymin=0.03,
ymax=200,
ytick={0.1,1,10},
yticklabels={{0.1},{  1},{ 10}},
yminorticks=true,
ylabel style={font=\color{white!15!black},font=\scriptsize},
ylabel={Performance [GFLOPs/s]},
axis background/.style={fill=white},
legend style={at={(0.97,0.03)}, anchor=south east, legend cell align=left, align=left, draw=white!15!black}
]

\addplot [color=gray, line width=0.8pt, forget plot, dashed]
table[row sep=crcr]{	0.13	36.44\\
	1.5	36.44\\
};
\addplot [color=gray, line width=0.8pt, forget plot, dashed]
table[row sep=crcr]{	0.015	4.56\\
	1.5	4.56\\
};

\addplot [color=black, line width=1.2pt, forget plot]
table[row sep=crcr]{	0.01	2.8984\\
	0.249758487441347	72.39\\
	1.5	72.39\\
};
\addplot [color=black, line width=1.2pt, forget plot]
table[row sep=crcr]{	0.01	0.913\\
	0.792880613362541	72.39\\
	1.5	72.39\\
};
\addplot [color=black, line width=1.2pt, forget plot]
table[row sep=crcr]{	0.01	0.2291\\
	3.15975556525535	72.39\\
	1.5	72.39\\
};
\addplot [color=black, line width=1.2pt, forget plot]
table[row sep=crcr]{	0.01	0.047\\
	15.4021276595745	72.39\\
	1.5	72.39\\
};

\addplot [color=color1, dashdotted, line width=1.2pt, forget plot]
table[row sep=crcr]{	0.01	3.8\\
	0.13	3.8\\
	0.13	0.03\\
};
\addplot [color=color1, line width=1.2pt, draw=none, mark size=4.3pt, mark=diamond*, mark options={solid, fill=color1, color1}]
table[row sep=crcr]{	0.13	3.8\\
};
\addlegendentry{Physical scaling}

\addplot [color=color2, dashed, line width=1.2pt, forget plot]
table[row sep=crcr]{	0.01	4.89\\
	0.19	4.89\\
	0.19	0.03\\
};
\addplot [color=color2, line width=1.2pt, draw=none, mark size=1.8pt, mark=square*, mark options={solid, fill=color2, color2}]
table[row sep=crcr]{	0.19	4.89\\
};
\addlegendentry{Nodal integration}

\addplot [color=color5, line width=1.2pt, forget plot]
table[row sep=crcr]{	0.01	0.99\\
	0.11	0.99\\
	0.11	0.03\\
};
\addplot [color=color5, line width=1.2pt, draw=none, mark size=2.5pt, mark=*, mark options={solid, fill=color5, color5}]
table[row sep=crcr]{	0.11	0.99\\
};
\addlegendentry{Stored stencil}

\node[right, align=left, rotate=11.5]
at (axis cs:0.011,4) {L1 Bandwidth 289.84 GB/sec};
\node[right, align=left, rotate=11.5]
at (axis cs:0.011,1.3) {L2 Bandwidth 91.30 GB/sec};
\node[right, align=left, rotate=11]
at (axis cs:0.011,0.34) {L3 Bandwidth 22.91 GB/sec};
\node[right, align=left, rotate=11]
at (axis cs:0.011,0.075) {DRAM Bandwidth 4.70 GB/sec};
\node[right, align=left]
at (axis cs:0.25,100.01) {DP Vector FMA 72.39 GFLOPS};

\node[right, align=left]
at (axis cs:0.25,28) {DP Vector Add 36.44 GFLOPS};

\node[right, align=left]
at (axis cs:0.25,6) {DP Scalar Add 4.56 GFLOPS};
\end{axis}
\end{tikzpicture}
		\caption{\label{fig:roofline}Roofline analysis of the residual computation using nodal integration, physical stencil scaling, and stored stencils approaches.}
\end{figure}
\subsubsection{Linear elastostatics benchmark problem}
As a first benchmark problem, we consider a compressible linear elasticity problem on the unit cube $\Omega = (0,1)^3$ modeled by \cref{eqn:strongform} with $\Gamma_\text{D} = \partial \Omega$ and $\Gamma_\text{N} = \emptyset$. The material of the block is assumed to be isotropic and heterogeneous with a varying elastic modulus $E$ but constant Poisson's ration $\nu$. In this scenario, the stress tensor $\stress = 2\mu \strain + \lambda \; \trace{\strain} \identity$ is given by Hooke's law and the Lam\'{e} constants $\mu$ and $\lambda$ are
\begin{align}
\mu(E) = \frac{E}{2(1+\nu)}, \;\; \text{and} \;\; \lambda(E) = \frac{\nu E}{(1+\nu)(1-2\nu)}.
\end{align}
Since the stress tensor $\stress$ depends linearly on $E$, we factor it out and rewrite the stress tensor such that it depends only on the single spatially variable coefficient $E$, i.e.,
\begin{align}
\label{eqn:stressE}
\stress = E(x,y,z) \cdot \left(\frac{1}{(1+\nu)}\strain  + \frac{\nu}{(1+\nu)(1-2\nu)} \trace{\strain} \identity\right).
\end{align}
The associated bilinear form in the \emph{constant case} $E = 1$ is thus given by a linear combination of the discussed forms, yielding
\begin{align}
\label{eqn:bilinearformE}
a^{E=1}(\v{u}, \v{v}) = \frac{1}{(1+\nu)} \scpd[\Omega]{\strain\left(\v{u}\right)}{\strain\left(\v{v}\right)} + \frac{\nu}{(1+\nu)(1-2\nu)} \scpd[\Omega]{\div \v{u}}{\div \v{v}}.
\end{align}
The scaling is then performed on the bilinear form \cref{eqn:bilinearformE} with $E$ as the varying scalar coefficient.
In the following, we perform a quantitative comparison of the three approaches by investigating their accuracy and run-time.
For this purpose, we let $\v{u}^\ast$ be a manufactured solution and set the right hand side $\v{f}$ of \cref{eqn:strongform} accordingly to $\v{f} = -\div \stress(\v{u}^\ast)$. The Dirichlet boundary condition is set to $\v{g} = \restr{\v{u}^\ast}{\partial \Omega}$. This allows for a direct computation of errors and a quantitative study on accuracy of the different methods.
By $\mathcal{I}_h$, we denote the interpolation operator of a function on the mesh $\mesh_h$ and by $\ltwonorm{\cdot}$ the discrete $L^2$ norm defined as
\begin{align}
\ltwonorm{\v{u}} = \left(h^3 \sum_{i \in \mathcal{N}_h} \eucnorm{\v{u}(\v{x}_i)}^2\right)^\frac{1}{2},
\end{align}
where $\mathcal{N}_h$ is the set of all vertices in  the mesh $\mesh_h$.
As material parameters, we choose the Poisson's ratio of Aluminum, i.e., $\nu = 0.34$, and a Young's modulus of the following form
\begin{align}
E(x,y,z) = \cos(m\, \pi\, x\, y\, z) + 2,\; m \in \{ 1,2,3,8 \}.
\end{align}
The manufactured solution $\v{u}^\ast$ is chosen as
\begin{align}
\v{u}^\ast(x,y,z) = \frac{1}{x y z + 1} \begin{pmatrix}
x^{3} y + z^{2}\\
x^{4} y + 2 z\\
3 x + y z^{3}
\end{pmatrix}.
\end{align}
It is important to note that the coefficient and exact solution do not lie in the ansatz spaces and therefore cannot be exactly reproduced.

We discretize the computational domain by $384$ tetrahedra on the coarsest level $\ell = 0$. The finest level considered in this subsection is $L = 6$. Each system of equations is solved using a single \revised{rank on SuperMUC-NG} and by employing a geometric $V(3,3)$ multigrid solver until a relative \revised{residual} of $10^{-8}$ is obtained.
\revised{As a smoother, we employ the hybrid Gauss-Seidel method and on the coarsest level, we employ a diagonally preconditioned conjugate gradient method since the problem is symmetric and positive definite.}

In \cref{tab:benchmarkphysical_and_unphysical}, we report on the errors, convergence rates, \revised{number of required V-cycle iterations}, and run-times for different refinement levels $\ell$ and coefficient parameters $m$. The error on level $\ell$ is defined as $\ltwonorm{\mathcal{I}_{h_L} \v{u}^\ast - \mathcal{I}_{h_L}\v{v}_{h_\ell}}$, where $\v{v}_{h_\ell}$ denotes the numerical solution obtained with one of the \revised{two matrix-free approaches, i.e., $\v{u}_h$ and $\hat{\v{u}}_h$.
We do not consider the stored stencil approach in this comparison, since the problem sizes on the finest level are too large and the matrices could not be stored in memory}.
We observe quadratic convergence in the discrete $L^2$ norm for the assembly through nodal integration and the physical scaling. \revised{On the last level $L = 6$, the convergence rate is higher since here we compare two discrete approximations on the same level.
We still kept these rows in the table for completeness and in order to compare the relative tts even if the error is not directly comparable to the errors on the coarser levels.
Independent of the coefficient frequency, we observe a relative tts of about $61\%$ for $m\in \{1,2,3,8\}$.}

\revised{In order to emphasize that using the unphysical scaling results in a wrong solution, we computed the errors in this benchmark using the unphysical scaling \cref{eqn:stencilscalingwithoutcorr}.
In this case the errors on level $\ell = 6$ were $\num{3.95e-3}$ for $m=1$, $\num{1.25e-2}$ for $m=2$, $\num{1.83e-2}$ for $m=3$, and $\num{2.32e-3}$ for $m=8$.}

\revised{For a performance comparison, we performed the same experiment on a compute node of the older SuperMUC Phase 2.
The number V-cycles and errors are the same as on SuperMUC-NG, thus we only present the tts for both matrix-free approaches in \cref{tab:benchmarkphase2}.
On this machine, we observe a relative tts of about $45\%$ independent of the coefficient frequency.
This larger speedup is due to the lower clock rate of the processor which has already been discussed in \Cref{sec:memorytrafficroofline} and illustrated in \cref{fig:performancecomparison}.}

\begin{table}[h]
	\centering\footnotesize
	\caption{\label{tab:benchmarkphysical_and_unphysical} Errors for the linear elastostatics benchmark problem in the discrete $L^2$ norm, convergence rates, \revised{number of V-cycle iterations}, time-to-solution and relative time-to-solution for nodal integration, physical scaling and unphysical scaling recorded for different refinement levels $\ell$ and parameters $m$. \revised{The measurements were conducted on SuperMUC-NG.}}
		\begin{tabular}{*5{c|}r|*3{c|}r|r}
		\toprule
		& & \multicolumn{4}{c|}{nodal integration} & \multicolumn{4}{c}{physical scaling} & \multicolumn{1}{c}{rel.}\\
		$\ell$ & DoFs & error & eoc & \revised{iter} & tts [s] & error & eoc & \revised{iter} & tts [s] & tts \\
		\midrule
		\multicolumn{11}{c}{$m=1$}\\
		\midrule
		1 & \num{1.52e+03} & \num{8.11e-03} & 0.00 & \revised{12} & \revised{   0.07}  & \num{8.07e-03} & 0.00 & \revised{12} & \revised{   0.07} & \revised{0.93} \\ 
		2 & \num{1.21e+04} & \num{2.12e-03} & 1.94 & \revised{18} & \revised{   0.35}  & \num{2.10e-03} & 1.94 & \revised{18} & \revised{   0.29} & \revised{0.81} \\ 
		3 & \num{9.72e+04} & \num{5.43e-04} & 1.97 & \revised{25} & \revised{   2.97}  & \num{5.39e-04} & 1.97 & \revised{25} & \revised{   2.03} & \revised{0.68} \\ 
		4 & \num{7.77e+05} & \num{1.38e-04} & 1.97 & \revised{29} & \revised{  26.50}  & \num{1.37e-04} & 1.97 & \revised{29} & \revised{  17.13} & \revised{0.65} \\ 
		5 & \num{6.22e+06} & \num{3.53e-05} & 1.97 & \revised{32} & \revised{ 230.78}  & \num{3.51e-05} & 1.97 & \revised{32} & \revised{ 143.04} & \revised{0.62} \\ 
		6 & \num{4.97e+07} & \num{4.71e-06} & 2.91 & \revised{32} & \revised{1848.95}  & \num{4.65e-06} & 2.92 & \revised{32} & \revised{1125.28} & \revised{0.61} \\
		\midrule
		\multicolumn{11}{c}{$m=2$}\\
		\midrule
		1 & \num{1.52e+03} & \num{8.26e-03} & 0.00 & \revised{12} & \revised{   0.08} & \num{8.23e-03} & 0.00 & \revised{12} & \revised{   0.07}  & \revised{0.93} \\ 
		2 & \num{1.21e+04} & \num{2.16e-03} & 1.93 & \revised{19} & \revised{   0.37} & \num{2.15e-03} & 1.94 & \revised{19} & \revised{   0.30}  & \revised{0.81} \\ 
		3 & \num{9.72e+04} & \num{5.54e-04} & 1.97 & \revised{25} & \revised{   2.96} & \num{5.50e-04} & 1.97 & \revised{25} & \revised{   2.02}  & \revised{0.68} \\ 
		4 & \num{7.77e+05} & \num{1.41e-04} & 1.97 & \revised{29} & \revised{  26.40} & \num{1.40e-04} & 1.97 & \revised{29} & \revised{  16.88}  & \revised{0.64} \\ 
		5 & \num{6.22e+06} & \num{3.59e-05} & 1.97 & \revised{32} & \revised{ 233.40} & \num{3.57e-05} & 1.97 & \revised{32} & \revised{ 142.98}  & \revised{0.61} \\ 
		6 & \num{4.97e+07} & \num{4.92e-06} & 2.87 & \revised{33} & \revised{1910.68} & \num{4.86e-06} & 2.87 & \revised{33} & \revised{1173.80}  & \revised{0.61} \\
		\midrule
		\multicolumn{11}{c}{$m=3$}\\
		\midrule
		1 & \num{1.52e+03} & \num{8.28e-03} & 0.00 & \revised{12} & \revised{   0.07} & \num{8.30e-03} & 0.00 & \revised{12} & \revised{   0.07}  & \revised{0.95} \\ 
		2 & \num{1.21e+04} & \num{2.16e-03} & 1.94 & \revised{19} & \revised{   0.37} & \num{2.17e-03} & 1.94 & \revised{19} & \revised{   0.30}  & \revised{0.81} \\ 
		3 & \num{9.72e+04} & \num{5.54e-04} & 1.97 & \revised{25} & \revised{   2.95} & \num{5.54e-04} & 1.97 & \revised{25} & \revised{   2.01}  & \revised{0.68} \\ 
		4 & \num{7.77e+05} & \num{1.41e-04} & 1.97 & \revised{30} & \revised{  27.58} & \num{1.41e-04} & 1.97 & \revised{29} & \revised{  16.87}  & \revised{0.61} \\ 
		5 & \num{6.22e+06} & \num{3.59e-05} & 1.97 & \revised{32} & \revised{ 230.73} & \num{3.59e-05} & 1.97 & \revised{32} & \revised{ 145.07}  & \revised{0.63} \\ 
		6 & \num{4.97e+07} & \num{4.99e-06} & 2.85 & \revised{33} & \revised{1902.81} & \num{5.04e-06} & 2.84 & \revised{33} & \revised{1160.50}  & \revised{0.61} \\
		\midrule
		\multicolumn{11}{c}{$m=8$}\\
		\midrule
		1 & \num{1.52e+03} & \num{8.67e-03} & 0.00 & \revised{11} & \revised{   0.07} & \num{9.04e-03} & 0.00 & \revised{11} & \revised{   0.06} & \revised{0.94} \\ 
		2 & \num{1.21e+04} & \num{2.26e-03} & 1.94 & \revised{18} & \revised{   0.36} & \num{2.44e-03} & 1.89 & \revised{18} & \revised{   0.29} & \revised{0.79} \\ 
		3 & \num{9.72e+04} & \num{5.75e-04} & 1.98 & \revised{24} & \revised{   2.83} & \num{6.36e-04} & 1.94 & \revised{24} & \revised{   1.96} & \revised{0.69} \\ 
		4 & \num{7.77e+05} & \num{1.46e-04} & 1.98 & \revised{28} & \revised{  25.69} & \num{1.63e-04} & 1.96 & \revised{28} & \revised{  16.46} & \revised{0.64} \\ 
		5 & \num{6.22e+06} & \num{3.72e-05} & 1.97 & \revised{31} & \revised{ 223.97} & \num{4.14e-05} & 1.98 & \revised{31} & \revised{ 138.66} & \revised{0.62} \\ 
		6 & \num{4.97e+07} & \num{5.54e-06} & 2.75 & \revised{32} & \revised{1844.60} & \num{6.90e-06} & 2.58 & \revised{32} & \revised{1124.61} & \revised{0.61} \\
		\bottomrule
		\noalign{\smallskip}
	\end{tabular}
\end{table}

\begin{table}[h]
\centering\footnotesize
\caption{\label{tab:benchmarkphase2} \revised{Time-to-solution and relative time-to-solution in the linear elastostatics benchmark for nodal integration and physical scaling recorded for different refinement levels $\ell$ and parameters $m$. The measurements were conducted on SuperMUC Phase 2.}}
\begin{tabular}{c|r|r|c}
\toprule
& \multicolumn{1}{c|}{nodal integration} & \multicolumn{1}{c}{physical scaling} & \multicolumn{1}{c}{rel.}\\
$\ell$ & tts [s] & tts [s] & tts\\
\midrule
\multicolumn{4}{c}{$m=1$}\\
\midrule
1 &    0.10 &    0.08 & 0.88\\ 
2 &    0.84 &    0.34 & 0.40\\ 
3 &    4.48 &    2.39 & 0.53\\ 
4 &   42.36 &   20.62 & 0.49\\ 
5 &  381.03 &  176.84 & 0.46\\ 
6 & 3065.13 & 1394.86 & 0.46\\
\midrule
\multicolumn{4}{c}{$m=2$}\\
\midrule
1 &    0.09 &     0.09 & 0.93\\ 
2 &    0.53 &     0.36 & 0.68\\ 
3 &    4.79 &     2.35 & 0.49\\ 
4 &   44.31 &    20.73 & 0.47\\ 
5 &  381.06 &   174.61 & 0.46\\ 
6 & 3170.62 &  1420.91 & 0.45\\
\bottomrule
\noalign{\smallskip}
\end{tabular}\hspace{1em}
\begin{tabular}{c|r|r|c}
\toprule
& \multicolumn{1}{c|}{nodal integration} & \multicolumn{1}{c}{physical scaling} & \multicolumn{1}{c}{rel.}\\
$\ell$ & tts [s] & tts [s] & tts\\
\midrule
\multicolumn{4}{c}{$m=3$}\\
\midrule
1 &    0.09 &    0.08 & 0.91\\ 
2 &    0.58 &    0.35 & 0.61\\ 
3 &    4.76 &    2.43 & 0.51\\ 
4 &   45.51 &   21.02 & 0.46\\ 
5 &  379.64 &  174.09 & 0.46\\ 
6 & 3159.07 & 1420.09 & 0.45\\
\midrule
\multicolumn{4}{c}{$m=8$}\\
\midrule
1 &    0.09 &    0.08 & 0.91\\ 
2 &    0.49 &    0.42 & 0.85\\ 
3 &    4.53 &    2.26 & 0.50\\ 
4 &   42.36 &   20.35 & 0.48\\ 
5 &  369.05 &  168.28 & 0.46\\ 
6 & 3079.40 & 1374.16 & 0.45\\
\bottomrule
\noalign{\smallskip}
\end{tabular}

\end{table}

\subsubsection{Linear elastostatics with external forces}
In this subsection, we present an application of our scaling approach where an external force is applied to an isotropic and heterogeneous material. As before, we consider the stress tensor of Hooke's law and model the problem by \cref{eqn:strongform}, where $\partial \Omega = \Gamma_\text{D} \cup \Gamma_{\text{N}}$ and $\Omega = (0,4)\times(0,2)\times(0,1)$, cf. left of Figure~\ref{fig:linearelasticitydomain}. The Dirichlet boundary is chosen as $\Gamma_\text{D} = \{ (x,y,z) \in \Omega \; | \; z = 0 \}$ and the Neumann boundary as $\Gamma_\text{N} = \partial\Omega \backslash \Gamma_\text{D}$. In this scenario, we ignore volume forces, thus we set $\v{f} = \v{0}$.
The material block is clamped at the bottom, therefore we set $\v{g} = \v{0}$. Further, the following planar force $\v{\hat{t}}$ is applied to the top plane of the foam
\begin{align}
\v{\hat{t}}(x,y,z) = \begin{cases} 
(0,0,-1)^\top & z = 1 \\
(0,0,0)^\top & \text{else}
\end{cases} \si{\giga\pascal}.
\end{align}

\begin{figure}[h]
\centering
\begin{minipage}{0.45\textwidth}
\begin{tikzpicture}[every edge quotes/.append style={auto}]
\tikzstyle{ground}=[fill,pattern=north east lines,draw=none,minimum width=0.75cm,minimum height=0.3cm]
\tikzstyle{ground_rot}=[fill,pattern=north east lines,draw=none,minimum width=0.75cm,minimum height=0.3cm, rotate=45]
\tikzstyle{load}   = [ultra thick,-latex]

\pgfmathsetmacro{\cubex}{4}
\pgfmathsetmacro{\cubey}{1}
\pgfmathsetmacro{\cubez}{2}
\draw [draw=black, every edge/.append style={draw=black, densely dashed, opacity=.5}]
(0,0,0) coordinate (o) -- ++(-\cubex,0,0) coordinate (a) -- ++(0,-\cubey,0) coordinate (b) edge coordinate [pos=1] (g) ++(0,0,-\cubez)  -- ++(\cubex,0,0) coordinate (c) -- cycle
(o) -- ++(0,0,-\cubez) coordinate (d) -- ++(0,-\cubey,0) coordinate (e) edge (g) -- (c) -- cycle
(o) -- (a) -- ++(0,0,-\cubez) coordinate (f) edge (g) -- (d) -- cycle;
\path [every edge/.append style={draw=black, |-|}]
(b) +(0,-12pt) coordinate (b1) edge ["4"'] (b1 -| c)
(b) +(-5pt,0) coordinate (b2) edge ["1"] (b2 |- a)
(c) +(5.5pt,-5.5pt) coordinate (c2) edge ["2"'] ([xshift=5.5pt,yshift=-5.5pt]e)
;

\fill [pattern = north west lines] (-\cubex, -\cubey) rectangle (0, -\cubey-0.3);

\pgfdeclarepatternformonly{west lines}{	\pgfqpoint{-1pt}{-1pt}}{\pgfqpoint{4pt}{4pt}}{\pgfqpoint{3pt}{3pt}}{
	\pgfsetlinewidth{0.4pt}
	\pgfpathmoveto{\pgfqpoint{0pt}{3pt}}
	\pgfpathlineto{\pgfqpoint{3pt}{0pt}}
	\pgfusepath{stroke}
}

\fill [pattern = west lines] (0, -\cubey, 0) -- ++(0, -0.3, 0) -- ++(0, 0, -\cubez) -- ++(0,0.3,0) -- cycle;

\foreach \x in {0,-0.25,...,-\cubex} {
	\foreach \z in {0,-1.0,...,-\cubez} {
		\draw[-latex] (\x,0.5,\z) -- ++(0,-0.5,0);
	}
}
\end{tikzpicture}
\end{minipage}\begin{minipage}{0.5\textwidth}
\includegraphics[width=0.85\textwidth]{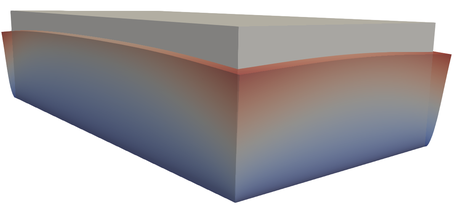}
\begin{tikzpicture}[scale=0.4]
\pgfmathsetlengthmacro\MajorTickLength{
  \pgfkeysvalueof{/pgfplots/major tick length} * 0.5
}
\begin{axis}[
title={\Large $\|u\|$},
xmin=0, xmax=0.02,
ymin=0, ymax=6e-2,
axis on top,
scaled x ticks=false,
scaled y ticks=false,
xtick=\empty,
xticklabels=\empty,
yticklabel pos=right,
y tick label style={
  /pgf/number format/.cd,
            sci zerofill,
            precision=1,
  /tikz/.cd  
},
extra y ticks={
      },
extra y tick style={
    tick style=transparent,     yticklabel pos=right,
    y tick label style={
        /pgf/number format/.cd,
            sci zerofill,
            precision=1,
      /tikz/.cd
    }
},
ylabel style={rotate=-90},
width=1.82cm,
height=5.5cm,
major tick length=\MajorTickLength
]
\addplot graphics [
includegraphics cmd=\pgfimage,
xmin=\pgfkeysvalueof{/pgfplots/xmin}, 
xmax=\pgfkeysvalueof{/pgfplots/xmax}, 
ymin=\pgfkeysvalueof{/pgfplots/ymin}, 
ymax=\pgfkeysvalueof{/pgfplots/ymax}
] {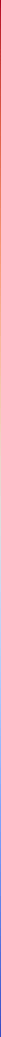};
\end{axis}
\end{tikzpicture}

\end{minipage}
\caption{\label{fig:linearelasticitydomain}Experimental setup and dimensions of the metal foam (left). Initial (gray) and displaced (colored) foam after applying the force on top. The displacement is magnified by a factor of 5. The numerical solution was computed on refinement level $\ell=4$ with standard nodal integration (right).}
\end{figure}
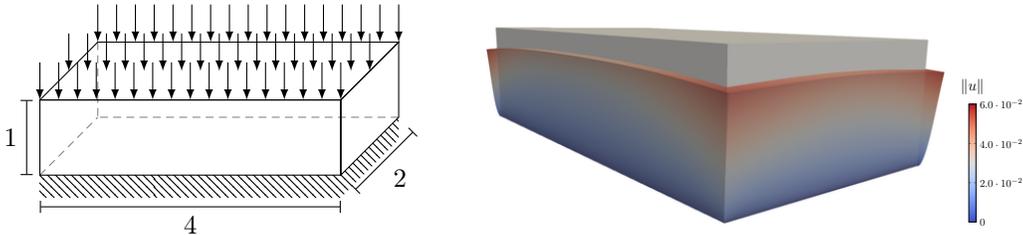
We assume that the material of interest is a metal foam, thus we apply the \emph{Gibson and Ashby model} \cite{zhang2003high, gibson1999cellular} which assumes the following relationship between the elastic modulus of the metal foam $E_f$ and of the matrix $E_m$
\begin{align}
\label{eq:gibsonashby}
\frac{E_f}{E_m} \approx \phi^2 \left(\frac{\rho_f}{\rho_m}\right)^2 + (1-\phi) \frac{\rho_f}{\rho_m},
\end{align}
where $\rho_f$ is the foam's density, $\rho_m$ the matrix density, and $\phi$ the porosity of the foam. Again, we assume the matrix to consist of Aluminum
with a Poisson's ratio of $\nu = 0.34$ and the elastic modulus $E_m = \SI{70}{\giga\pascal}$. Additionally, we assume that the ratio of foam- and matrix-density is given by a radially symmetric function of the form
\begin{align}
\frac{\rho_f}{\rho_m} = \frac{1}{16} x(4-x)z(2-y) + \frac{1}{2},
\end{align}
and $\phi = 1-\frac{\rho_f}{\rho_m}$. The foam's elastic modulus $E_f$ is then obtained by relationship \cref{eq:gibsonashby}.

We discretize the block with 3072 tetrahedra on the coarsest level $\ell = 0$. The finest level considered in this subsection is $L = 5$. Each system of equations is solved using 48 compute cores with the same multigrid solver as in the previous subsection.
\revised{Since no analytical solution is available, we assume that the solution obtained with the reference bilinear form $a_h(\cdot,\cdot)$ is the true solution and compare it to the solutions obtained using the form $\hat{a}_h(\cdot,\cdot)$. We denote the solutions by $\v{u}_h$ and $\hat{\v{u}}_h$, respectively.
Again, we do not consider the stored stencil approach in this comparison, since the problem sizes on the finest level are too large and the matrices could not be stored in memory.}
The error on level $\ell$ is defined by $\| \mathcal{I}_{h_L} \v{v}_{h_\ell} - \v{u}_{h_L} \|$ for \revised{$\v{v} \in \{\v{u}, \hat{\v{u}}\}$} and $\ell \leq L$.
See right of Figure~\ref{fig:linearelasticitydomain} for an illustration of the deformed metal foam computed on level $\ell = 4$.

In \cref{tab:foamphysical_and_unphysical}, we report on the errors, convergence rates, \revised{number of V-cycle iterations}, and run-times for different refinement levels $\ell$. We do not observe optimal quadratic convergence in the discrete $L^2$ norm, even in the nodal-integration case, because of the lower regularity of the problem. The solution obtained by the physical scaling, however, has the same convergence rate but with a relative tts of about \revised{$64\%$}.
\begin{table}[h]
	\centering\footnotesize
	\caption{\label{tab:foamphysical_and_unphysical} Errors of the linear elastostatics with external forces example in the discrete $L^2$ norm, convergence rates, \revised{number of V-cycle iterations}, time-to-solution, and relative time-to-solution for nodal integration, physical scaling recorded for different refinement levels $\ell$.}
	\begin{tabular}{*5{c|}r|*3{c|}r|r}
		\toprule
		& & \multicolumn{4}{c|}{nodal integration} & \multicolumn{4}{c}{physical scaling} & \multicolumn{1}{c}{rel.} \\
		$\ell$ & DoFs & error & eoc & \revised{iter} & tts [s] & error & eoc & \revised{iter} & tts [s] & \multicolumn{1}{c}{tts}  \\
		\midrule
		1 & \num{1.23e+04} & \num{4.46e-04} &  0.00 & \revised{12} &  \revised{ 0.38} & \num{4.44e-04} &  0.00 & \revised{12} &  \revised{ 0.42}  & \revised{1.11} \\
		2 & \num{9.86e+04} & \num{1.70e-04} &  1.39 & \revised{16} &  \revised{ 0.64} & \num{1.69e-04} &  1.39 & \revised{16} &  \revised{ 0.67}  & \revised{1.05} \\
		3 & \num{7.89e+05} & \num{6.45e-05} &  1.40 & \revised{18} &  \revised{ 1.27} & \num{6.42e-05} &  1.40 & \revised{18} &  \revised{ 1.20}  & \revised{0.94} \\
		4 & \num{6.31e+06} & \num{2.31e-05} &  1.48 & \revised{19} &  \revised{ 4.87} & \num{2.30e-05} &  1.48 & \revised{19} &  \revised{ 3.74}  & \revised{0.77} \\
		5 & \num{5.05e+07} & \num{6.59e-06} &  1.81 & \revised{19} &  \revised{30.02} & \num{6.57e-06} &  1.81 & \revised{19} &  \revised{19.22}  & \revised{0.64} \\
		\bottomrule
		\noalign{\smallskip}
	\end{tabular}
\end{table}
\subsection{Generalized incompressible Stokes problem}
\label{sec:genstokes}
In order to show that the new approach is also applicable to indefinite problems, we consider a generalized incompressible Stokes problem with a variable viscosity.
The stress tensor of a generalized Newtonian fluid with viscosity $\mu$ is given by $\stress(\v{u},p) = 2\mu \strain(\v{u})  - p \identity$ and depends not only on the velocity $\v{u}$ but additionally on the pressure $p$.
The problem considered in this section is modeled by the following equations
\begin{equation}
\label{eqn:generalizedstokes}
\begin{alignedat}{2}
-\div \stress &= \v{f} &&\quad \text{in } \Omega,\\
\div \v{u} &= 0 &&\quad \text{in } \Omega,\\
\v{u} &= \v{g} &&\quad \text{on } \Gamma_{\text{D}},\\
\stress \cdot \normal &= \v{\hat{t}} &&\quad \text{on } \Gamma_{\text{N}}.
\end{alignedat}
\end{equation}
on a domain $\Omega \subset \reals^3$ with a Dirichlet boundary $\Gamma_{\text{D}}$ and Neumann boundary $\Gamma_{\text{N}}$.
For the well posedness of the problem, the \revised{finite element} spaces need to meet a uniform inf-sup condition which is not the case for an equal-order $P1$ discretization. Therefore, we add a level dependent residual based stabilization term $c_\ell$ \cite{brezzi1984stabilization} to the mass conservation equation, i.e., $c_\ell(p,q) = -\frac{h_\ell^2}{12} \scpd[\Omega]{\grad{p}}{\grad{q}}$. If $\Gamma_{\text{N}} = \emptyset$ then the pressure is not unique up to a constant and we enforce uniqueness by demanding that the mean value of the pressure is zero.
\revised{This equal order discretization is inconsistent and does not conserve the mass locally, however, the discrete solutions still convergence with the optimal order.
There are possibilities to obtain local mass conservation by a post process which is discussed in \cite{waluga2006masscorrection}.}
\subsubsection{Stationary geophysics example}
To demonstrate that the presented method is also suitable for solving geophysical problems, we present an example inspired by convection in the Earth's mantle.
The domain is chosen as $\Omega = (0,1)^3$ with $\Gamma_{\text{D}} = \partial\Omega$ and $\Gamma_{\text{N}} = \emptyset$.
In this scenario, the viscosity and the volume forces depend on the temperature. Therefore, we construct a temperature field $\vartheta$, resembling a temperature plume in the Earth's mantle given by following formula 
\begin{align}
\vartheta(x,y,z) = \frac{89\, \mathrm{e}^{ - 30\, {\left(z + \left(\frac{3\, r}{2} + \frac{3}{4}\right)\, \left(r - \frac{1}{2}\right) - \frac{3}{10}\right)}^2 - 10\, r^2}}{100} + \frac{49\, \mathrm{e}^{- 100\, r^2}}{50\, \left(\mathrm{e}^{17\, z - \frac{1819}{200}} + 1\right)},
\end{align}
with $r(x,y,z) = \sqrt{\frac{13\, {\left(x - \frac{1}{2}\right)}^2}{10} + \frac{27\, {\left(y - \frac{1}{2}\right)}^2}{10}}$. Note that the temperature field is not radially symmetric and therefore no problem reduction due to symmetry is possible. The viscosity $\mu$ of the fluid is then given by an exponential law \revised{with a jump across a horizontal plane}, i.e.,
\revised{
\begin{align}
\mu(x,y,z) = \mathrm{e}^{-\vartheta(x,y,z)} \cdot \begin{cases}
10^{-2} & z > \frac{3}{4} \\
1 & \text{else}
\end{cases}.
\end{align}
See left of \cref{fig:plume} for an illustration.
In cases like this, where the location of a jump is known a priori, it is possible to resolve the jump via the macro mesh, since the standard on-the-fly integration is performed across these interfaces.
If the locations of jumps are not known beforehand, it is still possible to locally mark elements where the standard on-the-fly integration should be performed.
This solution will of course reduce the performance because of the additional branches in the code and possible load imbalances between processes.}
\revised{Stokes flow problems with larger viscosity jumps or a highly heterogeneous viscosity require more sophisticated preconditioners than the geometric multigrid solver used here.
See, e.g.,~\cite{rudi2017weighted} how to construct a robust iterative solver in this case}.
Additionally, we assume a gravitational source term $\v{f} = \vartheta \cdot (0, 0, 10)^\top$ arising from a Boussinesq approximation \cite{waluga2006masscorrection,ricard2007physics}. \cref{fig:plume} on the right shows the velocity streamlines of the numerical solution using nodal integration computed on a mesh with $50\,331\,648$ tetrahedra.
\begin{figure}[h]
\centering
\begin{minipage}{\linewidth}
\includegraphics[keepaspectratio,height=3.7cm]{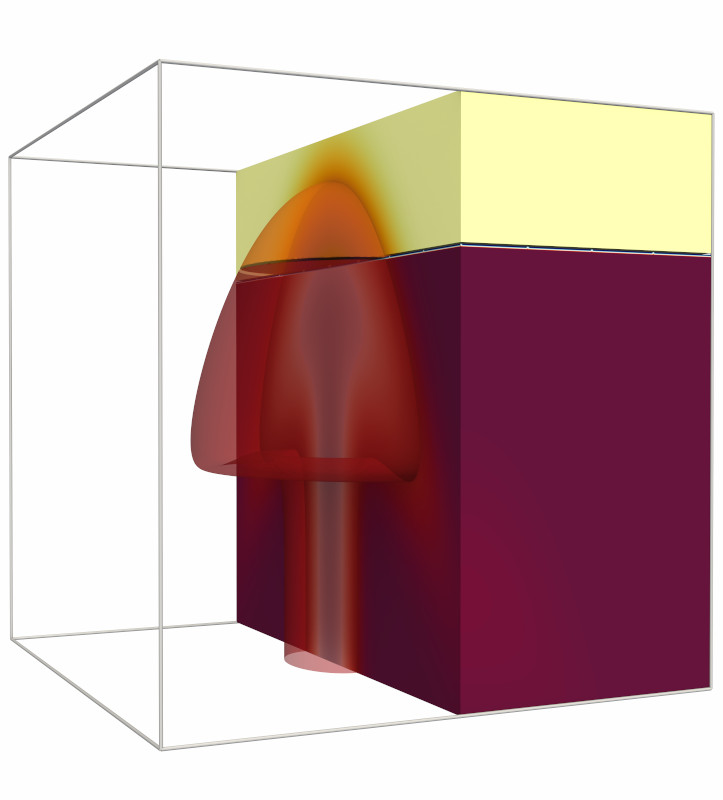}\begin{tikzpicture}[scale=0.4]
\pgfmathsetlengthmacro\MajorTickLength{
  \pgfkeysvalueof{/pgfplots/major tick length} * 0.5
}
\begin{axis}[
title={\Large $\mu$},
xmin=0, xmax=0.02,
ymin=4.7e-3, ymax=1e-2,
axis on top,
scaled x ticks=false,
scaled y ticks=false,
xtick=\empty,
xticklabels=\empty,
yticklabel pos=right,
y tick label style={
  /pgf/number format/.cd,
            sci zerofill,
            precision=1,
  /tikz/.cd  
},
extra y ticks={
   \pgfkeysvalueof{/pgfplots/ymin}
   },
extra y tick style={
    tick style=transparent,     yticklabel pos=right,
    y tick label style={
        /pgf/number format/.cd,
            sci zerofill,
            precision=1,
      /tikz/.cd
    }
},
ylabel style={rotate=-90},
width=1.82cm,
height=4.5cm,
major tick length=\MajorTickLength
]
\addplot graphics [
includegraphics cmd=\pgfimage,
xmin=\pgfkeysvalueof{/pgfplots/xmin}, 
xmax=\pgfkeysvalueof{/pgfplots/xmax}, 
ymin=\pgfkeysvalueof{/pgfplots/ymin}, 
ymax=\pgfkeysvalueof{/pgfplots/ymax}
] {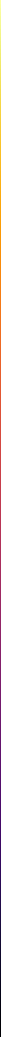};
\end{axis}
\begin{axis}[
yshift=-4.5cm,
title={\Large $\mu$},
xmin=0, xmax=0.02,
ymin=4.7e-3, ymax=1,
axis on top,
scaled x ticks=false,
scaled y ticks=false,
xtick=\empty,
xticklabels=\empty,
yticklabel pos=right,
y tick label style={
  /pgf/number format/.cd,
            sci zerofill,
            precision=1,
  /tikz/.cd  
},
extra y ticks={
   \pgfkeysvalueof{/pgfplots/ymin}
   },
extra y tick style={
    tick style=transparent,     yticklabel pos=right,
    y tick label style={
        /pgf/number format/.cd,
            sci zerofill,
            precision=1,
      /tikz/.cd
    }
},
ylabel style={rotate=-90},
width=1.82cm,
height=4.5cm,
major tick length=\MajorTickLength
]
\addplot graphics [
includegraphics cmd=\pgfimage,
xmin=\pgfkeysvalueof{/pgfplots/xmin}, 
xmax=\pgfkeysvalueof{/pgfplots/xmax}, 
ymin=\pgfkeysvalueof{/pgfplots/ymin}, 
ymax=\pgfkeysvalueof{/pgfplots/ymax}
] {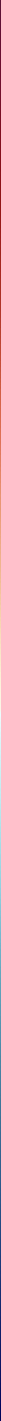};
\end{axis}
\end{tikzpicture}
\includegraphics[keepaspectratio,height=3.7cm]{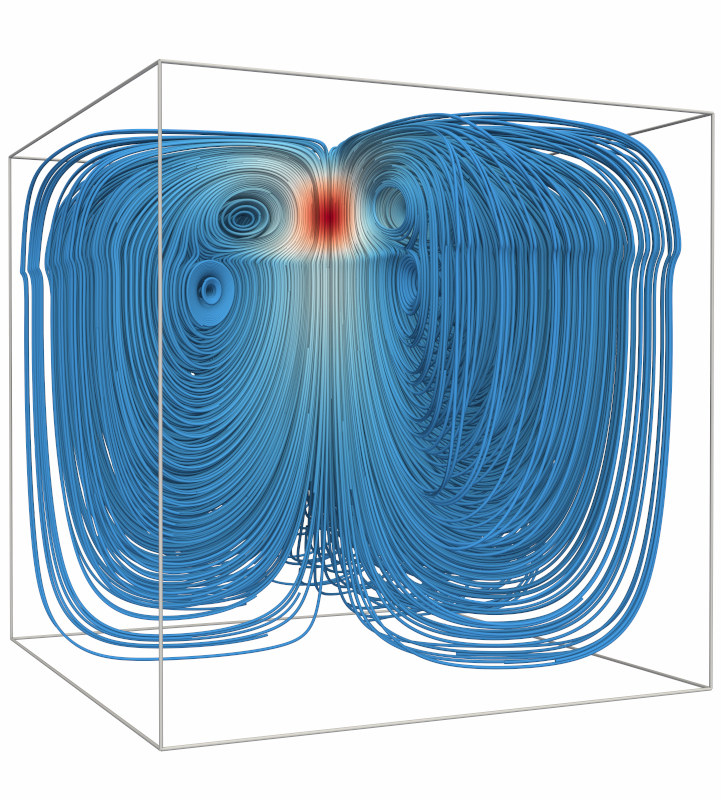}
\begin{tikzpicture}[scale=0.4]
\pgfmathsetlengthmacro\MajorTickLength{
  \pgfkeysvalueof{/pgfplots/major tick length} * 0.5
}
\begin{axis}[
title={\Large $\|u\|$},
xmin=0, xmax=0.02,
ymin=0, ymax=3.5e-1,
axis on top,
scaled x ticks=false,
scaled y ticks=false,
xtick=\empty,
xticklabels=\empty,
yticklabel pos=right,
y tick label style={
  /pgf/number format/.cd,
            sci zerofill,
            precision=1,
  /tikz/.cd  
},
extra y ticks={
      \pgfkeysvalueof{/pgfplots/ymax}
},
extra y tick style={
    tick style=transparent,     yticklabel pos=right,
    y tick label style={
        /pgf/number format/.cd,
            sci zerofill,
            precision=1,
      /tikz/.cd
    }
},
ylabel style={rotate=-90},
width=1.82cm,
height=5.5cm,
major tick length=\MajorTickLength
]
\addplot graphics [
includegraphics cmd=\pgfimage,
xmin=\pgfkeysvalueof{/pgfplots/xmin}, 
xmax=\pgfkeysvalueof{/pgfplots/xmax}, 
ymin=\pgfkeysvalueof{/pgfplots/ymin}, 
ymax=\pgfkeysvalueof{/pgfplots/ymax}
] {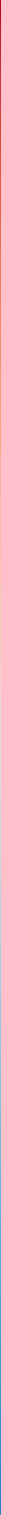};
\end{axis}
\end{tikzpicture}
\includegraphics[height=3.5cm]{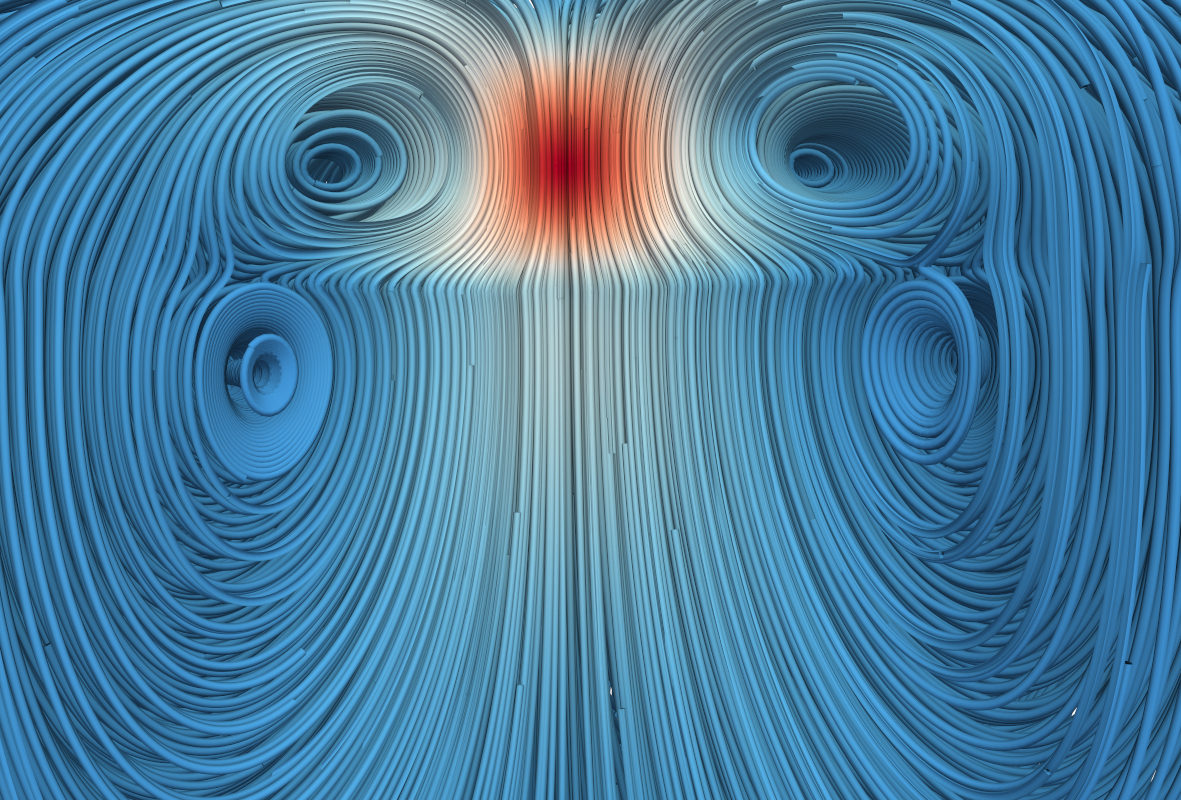}
\end{minipage}
\caption{\label{fig:plume} \revised{Viscosity $\mu$ depending on the given plume temperature field $\vartheta$ with isosurface $\mu = 0.85$ in the lower part and $\mu = 0.0085$ in the upper part (left). Velocity streamlines of the numerical solution computed on a mesh with $50\,331\,648$ tetrahedra using nodal integration (middle). Zoom on the velocity streamlines at the center (right).}}
\end{figure}
In the following scenario, the coarsest level $\ell = 0$ is discretized by $786\,432$ tetrahedra and each system is solved using $12\,288$ compute cores \revised{on SuperMUC-NG}.
\revised{The finest level $\ell=6$ involved solving a system with about \num{1.03e+11} DoFs.
In order to solve the systems, we employ the inexact Uzawa solver presented in \cite{drzisga2018analysis} with variable $V(3,3)$ cycles where two smoothing steps are added to each coarser refinement level which enforces convergence of the method.
As a smoother, we again employ the hybrid Gauss-Seidel method but with a relaxation parameter of $0.3$ in the pressure part.
On the coarsest level, we employ the diagonally preconditioned MINRES method since the problem is not positive definite but symmetric.}
\revised{The solutions obtained with both approaches do not show any visual differences on all levels.
In \cref{tab:stokesnodal}, we report on the number of inexact Uzawa iterations and tts for different refinement levels $\ell$. The relative tts of the physical scaling is based on the tts of the nodal integration.
}
The worse relative tts in comparison to the linear elasticity examples is due to the fact that the cost of the divergence matrices and the stabilization matrix need to be taken into account.
Since the stencils of these matrices are constant on each macro primitive and do not depend on a coefficient, we employ specialized kernels for them in all of the approaches.
\revised{This part of the cost remains unchanged for both approaches.
Therefore, assuming an optimal relative tts of $61$\% for the velocity block results in a theoretical optimal relative tts of only at about $78\%$ for the total Stokes operator.
This value is close to the value observed for $\ell = 6$ in \cref{tab:stokesnodal}.
We observe that the number of V-cycle iterations required to obtain the relative residual tolerance is larger for multigrid hierarchies with fewer levels compared to a larger number of levels.
Furthermore, the tts for the levels $\ell = 1$ to $\ell = 4$ are very close to each other.
Since the number of MPI ranks is fixed and the coarse grid solver does not scale optimally, most of the time is spent on the coarse grid if the number of multigrid levels is small.
This cost of the coarse grid solver decreases relatively if the number of multigrid levels increases and more time is spent for smoothing on the finer levels.}
\begin{table}[h]
	\centering\footnotesize
	\caption{\label{tab:stokesnodal}\revised{Velocity and pressure errors of the stationary geophysics example in the discrete $L^2$ norm, convergence rates, number of inexact Uzawa iterations, time-to-solution, and relative time-to-solution for nodal integration and physical scaling recorded for different refinement levels~$\ell$.}}
	\begin{tabular}{c|l|r|r|r|r|r}
       \toprule
       & & \multicolumn{2}{c|}{nodal integration} & \multicolumn{2}{c}{physical scaling} & \multicolumn{1}{c}{rel.}\\
       $\ell$ & \multicolumn{1}{c|}{DoFs} & iter &  tts [s] & iter & tts [s] & \multicolumn{1}{c}{tts} \\
       \midrule
       1 & \num{4.19e+06} & \revised{\num{12}} & \revised{ 36.58} & \revised{\num{12}} & \revised{ 39.30} & \revised{1.07} \\ 
       2 & \num{3.35e+07} & \revised{\num{11}} & \revised{ 32.58} & \revised{\num{11}} & \revised{ 34.98} & \revised{1.07} \\ 
       3 & \num{2.68e+08} & \revised{\num{10}} & \revised{ 31.69} & \revised{\num{10}} & \revised{ 32.22} & \revised{1.02} \\ 
       4 & \num{2.15e+09} & \revised{\num{9}}  & \revised{ 32.60} & \revised{\num{9}}  & \revised{ 32.38} & \revised{0.99} \\ 
       5 & \num{1.72e+10} & \revised{\num{9}}  & \revised{ 62.84} & \revised{\num{9}}  & \revised{ 54.91} & \revised{0.87} \\ 
       6 & \num{1.03e+11} & \revised{\num{8}}  & \revised{262.74} & \revised{\num{8}}  & \revised{197.58} & \revised{0.75} \\ 
       \bottomrule
       \noalign{\smallskip}
   \end{tabular}
\end{table}
\subsubsection{Non-linear generalized Stokes problem}
In this section, we consider the scenario of a non-linear incompressible Stokes problem where the fluid is assumed to be of generalized Newtonian type, modeled by a shear-thinning Carreau model
\begin{align}\label{eq:carreau}
\mu(\v{u}) = \eta_\infty + \left(\eta_0 - \eta_\infty\right) \left(1 + \kappa  |\strain(\v{u})|^2 \right)^{r}.
\end{align}
The considered parameters in dimensionless form are specified in left of \cref{fig:carreau_and_domain}. These parameters stem from experimental results, cf.~\cite[Chapter~II]{galdi2008hemodynamical}.
\begin{figure}[h]
\centering
\begin{minipage}[c]{0.35\textwidth}
\begin{tabular}{c | r}
\toprule
Parameter & \multicolumn{1}{c}{Value} \\
\midrule
$\eta_0$ & $140.764$ \\
$\eta_\infty$ & $1.0$ \\
$\kappa$ & $212.2$ \\
$r$ & $-0.325$\\
\bottomrule 
\end{tabular}
\end{minipage}
\begin{minipage}[c]{0.6\textwidth}
\includegraphics[width=\textwidth]{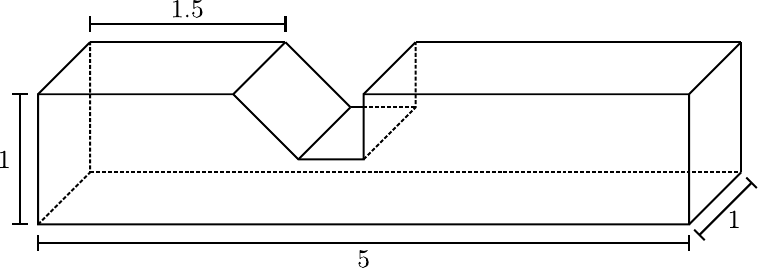}
\end{minipage}
\caption{\label{fig:carreau_and_domain}Dimensionless parameters for the Carreau viscosity model (left). Experimental setup and dimensions of channel (right).}
\end{figure}
The computational domain $\Omega$ is depicted in the right of \cref{fig:carreau_and_domain}, discretized by $14\,208$ tetrahedra on the coarsest level $\ell = 0$. The boundary $\partial \Omega$ is composed into Dirichlet and Neumann parts, i.e., $\partial \Omega = \Gamma_\text{D} \cup \Gamma_{\text{N}}$ with $\Gamma_\text{D} = \{ (x,y,z) \in \partial \Omega \; | \; x < 5 \}$ and $\Gamma_\text{N} = \partial\Omega \backslash \Gamma_\text{D}$.
The volume force term $\v{f}$, the external forces $\v{\hat{t}}$, and the Dirichlet boundary term $\v{g}$ are set to
\begin{align}
\v{f} = (0,0,100)^\top, \; \v{\hat{t}} = (0,0,0)^\top, \; \text{and} \;  \v{g} = 16 \, y \, (1 - y) \, z \, (1 - z)\cdot(1, 0, 1)^\top.
\end{align}
We solve this non-linear system by applying an inexact fixed-point iteration similar to the non-linear solver described in \cite{dembo1982inexact}, where the underlying linear systems are only solved approximately to prevent over solving. The pseudo-code of our approach is presented in \cref{alg:fixedpoint_mg}. The same inexact Uzawa multigrid solver described in the previous subsection is used for the computations in this subsection.
\revised{The problem is solved using a total of $111$ MPI ranks with 3 compute nodes on SuperMUC-NG.}
Following the standard notation, the discretized saddle-point problem in a single fixed-point iteration reads
\begin{align}
\label{eq:fixpointsystem}
\begin{pmatrix}
\sfA\left({\mu^{(n)}}\right) & \sfB^\top\\
\sfB & \sfC
\end{pmatrix} \begin{pmatrix}
\sfu^{(n+1)}\\
\sfp^{(n+1)}
\end{pmatrix} = \begin{pmatrix}
\sff\\
\sfzero
\end{pmatrix}.
\end{align}
\begin{algorithm}\caption{Fixed-point iterations coupled with a multigrid solver \label{alg:fixedpoint_mg}}
\begin{algorithmic}
\STATE Set $\sfu^{(0)} = \sfzero$, $\sfp^{(0)} = \sfzero$, $n = 0$
\STATE Set $\mu^{(0)} = \mu(\sfu^{(0)})$  by employing the local approximation from \Cref{sec:nodalcoeff}
\REPEAT
\STATE Solve system \cref{eq:fixpointsystem} for $\sfu^{(n+1)}$ and $\sfp^{(n+1)}$ by applying an inexact Uzawa V(3,3)-cycle
\STATE Set $\mu^{(n+1)} = \mu\left(\sfu^{(n+1)}\right)$ by employing the local approximation from \Cref{sec:nodalcoeff}
\STATE Set $n = n+1$
\UNTIL{$\max_{i}\{\| \sfu^{(n)}_i - \sfu^{(n-1)}_i \|_2\} < 10^{-3} \cdot \max_{i}\{\|\sfu_i^{(n)}\|_2\}$}
\STATE Solve system \cref{eq:fixpointsystem} for $\sfu^{(n+1)}$ and $\sfp^{(n+1)}$ by applying \revised{final} Uzawa V(3,3)-cycles until a relative residual of $10^{-3}$ is obtained
\STATE Set $\mu^{(n+1)} = \mu\left(\sfu^{(n+1)}\right)$  by employing the local approximation from \Cref{sec:nodalcoeff}
\RETURN $\sfu^{(n+1)}$, $\sfp^{(n+1)}$, and $\mu^{(n+1)}$
\end{algorithmic}
\end{algorithm}
In \cref{fig:stokes_profile}, we plot the final viscosity profile and y-component of the velocity along the line $\theta = [0,5] \times \{0.5\} \times \{0.42\}$
for \revised{both} approaches computed on $\ell = 4$. We see that the solutions of the nodal integration and physical scaling approaches coincide.
\revised{For the sake of completeness, we also present the unphysical scaling results which yield different curves in both figures.}
\begin{figure}
\centering
\begin{tikzpicture}[font=\scriptsize]\begin{axis}[
xlabel={$x$},
xmajorgrids,
ylabel={$\mu$},
ymajorgrids,
mark repeat=50,
height=4cm,
width=\textwidth,
]
\addplot[color5,mark=*,thick] table [x="arc_length", y="coeff_scaled_inexact", col sep=comma] {data/csv/stokes_fixedpoint_upwardinflow.csv};
\addplot[color2,mark=square*,thick] table [x="arc_length", y="coeff_fe", col sep=comma] {data/csv/stokes_fixedpoint_upwardinflow.csv};
\addplot[color1,mark=diamond*,thick] table [x="arc_length", y="coeff_scaled_exact", col sep=comma] {data/csv/stokes_fixedpoint_upwardinflow.csv};
\end{axis}
\end{tikzpicture}
\begin{tikzpicture}[font=\scriptsize]\centering
\begin{axis}[
xlabel={$x$},
xmajorgrids,
ylabel={$\|u\|$},
ymajorgrids,
mark repeat=50,
height=4cm,
width=\textwidth,
]
\addplot[color5,mark=*,thick] table [x="arc_length", y="mag_u_unphysical", col sep=comma] {data/csv/stokes_fixedpoint_upwardinflow.csv};
\addlegendentry{unphysical scaling};
\addplot[color2,mark=square*,thick] table [x="arc_length", y="mag_u_fe", col sep=comma] {data/csv/stokes_fixedpoint_upwardinflow.csv};
\addlegendentry{nodal integration};
\addplot[color1,mark=diamond*,thick] table [x="arc_length", y="mag_u_physical", col sep=comma] {data/csv/stokes_fixedpoint_upwardinflow.csv};
\addlegendentry{physical scaling};
\end{axis}
\end{tikzpicture}
\caption{\label{fig:stokes_profile}Viscosity profile (top) and profile of the velocity magnitude (bottom) plotted over $\theta$ for different assembly approaches.}
\end{figure}
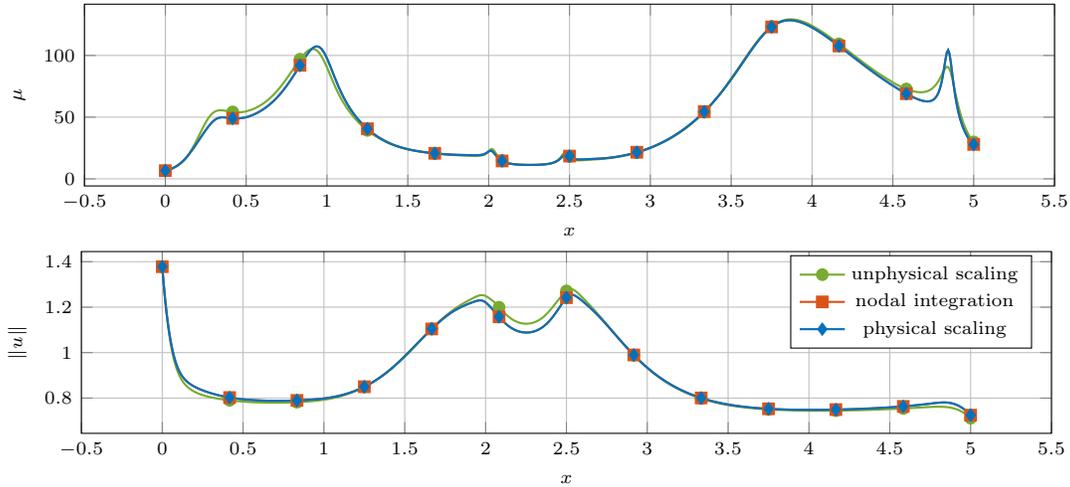

\revised{In \cref{tab:ttsnonlinearstokes}, we report on the number of fixed point iterations, the number of final iterations, and the absolute and relative time-to-solutions for the nodal integration and physical scaling.
} \revised{We observe a relative tts of about $87\%$ on the finest level $\ell = 6$.}
\begin{table}	\centering\footnotesize
	\caption{\label{tab:ttsnonlinearstokes} Relative time-to-solution comparison of the nodal integration and physical scaling approach for the non-linear generalized Stokes problem.}
	\begin{tabular}{c|c|c|c|r|c|c|r|r}
		\toprule
		& & \multicolumn{3}{c|}{nodal integration} & \multicolumn{3}{c}{physical scaling} & \multicolumn{1}{c}{rel.} \\
		$\ell$ & DoFs & \revised{fixed-point iter} & \revised{final iter} & \multicolumn{1}{c|}{tts [s]} & \revised{fixed-point iter} & \revised{final iter} & \multicolumn{1}{c|}{tts [s]} & \multicolumn{1}{c}{tts}  \\
		\midrule
		3 & \num{4.69E+06} & \revised{26} & \revised{14} & \revised{  42.50} & \revised{26} & \revised{14} & \revised{  41.80} & \revised{0.98}\\
		4 & \num{3.82E+07} & \revised{29} & \revised{7}  & \revised{ 127.98} & \revised{29} & \revised{7}  & \revised{ 121.87} & \revised{0.95}\\ 
		5 & \num{3.08E+08} & \revised{25} & \revised{8}  & \revised{ 571.84} & \revised{29} & \revised{8}  & \revised{ 578.57} & \revised{1.01}\\ 
		6 & \num{2.47E+09} & \revised{24} & \revised{7}  & \revised{3349.39} & \revised{25} & \revised{6}  & \revised{2925.09} & \revised{0.87}\\
		\bottomrule
		\noalign{\smallskip}
	\end{tabular}
\end{table}
Since the coefficient $\mu$ changes after each multigrid V-cycle, the caching of face stencils as it was done in the previous sections is not possible.
This has a large impact on the run-time for lower levels in the hierarchy, since the cost may be dominated by the face primitives.
Only asymptotically, for fine levels, the cost of the face primitives is small compared to the cost of the element primitives.
In this numerical experiment, the solver performance is worse than in the previous examples because of the inherent difficulty of the non-linear problem and the expensive on-the-fly nodal integration of the bilinear form on the \revised{macro faces}.

\section{Conclusion}
\label{sec:conclusion}
\revised{We have presented a new method to improve the performance of matrix-free operator applications for vector-valued second-order elliptic PDEs.	Although we restricted ourselves to linear finite elements on hierarchical hybrid grids, the idea of exploiting structure and symmetry in the mesh
 applies for higher order as well, which we described in a remark.
The method is based on scaling reference stencils originating from a constant coefficient discretization by variable coefficients.
We showed in theory that a correction term is required in cases where the coefficient is not constant.
Furthermore, we presented how to pre-compute these correction terms in case of HHGs and how to re-scale them using the coefficient.
This new approach was aimed to reduce memory traffic and to reduce the number of required floating point operations.
In order to show this, we first derived theoretical models about the required number of floating-point operations and the memory traffic.
We validated these estimates using benchmarks and numerical experiments on the supercomputer SuperMUC-NG.
The results have shown that specially designed matrix-free methods like our method may be beneficial compared to matrix-based methods not only for higher-order discretizations but also for low-order discretizations where the arithmetic intensity is lower.
The numerical benchmarks and experiments involving linear elasticity and Stokes flow also showed that our method is faster and comparably accurate as standard methods.
Although we only applied this method using a geometric multigrid solver it can also be used with more scalable coarse grid solvers like algebraic multigrid or combined with preconditioners designed for large viscosity jumps or heterogeneous viscosities in Stokes flow problems.
This makes our method attractive for highly resolved simulations on current and future machines where the available memory is limited compared to the compute power.
}

\section*{Acknowledgments}
This work was partly supported by the German Research Foundation through
the Priority Programme 1648 ``Software for Exascale Computing'' (SPPEXA) and by grant WO671/11-1.
The authors gratefully acknowledge the Gauss Centre for Supercomputing e.V. (GCS, \href{www.gauss-centre.eu}{www.gauss-centre.eu}) for funding this project by providing computing time on the GCS supercomputer SuperMUC at Leibniz Supercomputing Centre (LRZ, \href{www.lrz.de}{www.lrz.de}). 

\bibliographystyle{siam}
\bibliography{references}
\end{document}